\def\OPTIONArxiv{1}
\def\OPTIONAppendix{1}

\def\OPTIONConf{1}
\ifnum\OPTIONConf=1%
  \ifnum\OPTIONArxiv=0%
     \documentclass[reprint,nonatbib]{sigplanconf}
     \setpagenumber{17}   %
  \else
     \documentclass[reprint,nonatbib]{sigplanconf}
  \fi
\else
\documentclass[10pt,leqno]{article}
\fi
\usepackage[overload]{textcase}
\ifnum\OPTIONConf=1%
    \newcommand{\keywordfontsize}{9pt}  %
\else%
   \usepackage[headings]{fullpage}
   \newcommand{\keywordfontsize}{9.7pt}  %
\fi
\usepackage{color}
\usepackage[breaklinks,dvips]{hyperref}
\usepackage{breakurl}
\usepackage{pstricks,pst-node,graphicx,eepic}

\usepackage{listings}
\usepackage[authoryear]{natbib}
\bibpunct{(}{)}{;}{a}{}{,}

\usepackage{pgf,tikz}
\usetikzlibrary{shapes,arrows,matrix,fit,backgrounds,calc,decorations,decorations.pathmorphing}

\usepackage{xspace}
\ifnum\OPTIONConf=1%
\else
    \usepackage{charter}
\fi
\usepackage{euler}
\renewcommand{\ttdefault}{cmtt}   %

\usepackage{amsmath,amssymb,amsthm}

\usepackage{multirow}
\usepackage{stmaryrd}
\usepackage{pifont}
\usepackage{enumerate}
\usepackage{comment}
\usepackage{mathpartir} \usepackage{proof}
\usepackage{framed}
\usepackage{rotating}

\let\MathRightArrow\Rightarrow %
\usepackage{marvosym} %
\def\Rightarrow{\MathRightArrow}

\newtheorem{theorem}{Theorem}

\newtheorem{lemma}[theorem]{Lemma}
\newtheorem{corollary}[theorem]{Corollary}
\newtheorem*{corollary*}{Corollary}

\newtheorem*{conjecture*}{Conjecture}
\newtheorem{exercise}[theorem]{Exercise}
\newtheorem{example}[theorem]{Example}
\newtheorem*{example*}{Example}

\newtheorem*{theorem*}{Theorem}
\theoremstyle{remark} 
\newtheorem*{remark*}{Remark}
\theoremstyle{definition} \newtheorem{definition}[theorem]{Definition}

\usepackage{srcltx}

\newcommand{\ditto}{\ensuremath{''}}

\newcommand{\xLam}{\lambda}
\newcommand{\Lam}[1]{\xLam#1.\,}
\newcommand{\Pair}[2]{\texttt{(}#1\texttt{,}\;#2\texttt{)}}
\newcommand{\xFix}{\ensuremath{\keyword{fix}}}
\newcommand{\Fix}[1]{\xFix~#1.\:}

\newcommand{\sect}{\cap}

\renewcommand{\bfdefault}{b}             %

\newcommand{\arr}{\rightarrow}
\newcommand{\entails}{\,\vdash\,}         %

\newcommand{\subtype}{\mathrel{\leq}}       %
\newcommand{\sectty}{\mathrel{\mathcolor{dBlue}{\land}}}          %
\newcommand{\unty}{\mathrel{\mathcolor{dRed}{\lor}}}          %
\newcommand{\topty}{\top}

\newcommand{\step}{\mapsto}

\newcommand{\stepe}{\rightsquigarrow}
\newcommand{\stepse}{\stepe^*}

\newcommand{\steptm}{\step}
\newcommand{\stepstm}{\steptm^*}
\newcommand{\steptms}{\stepstm}

\newcommand{\xCase}{\ensuremath{\keyword{case}}}
\newcommand{\Case}[2]{{\xCase~#1~\keyword{of}~#2}}

\newcommand{\Casesumx}[3]{\Case{#1}{\Match{\Inl{#2}}{#3}}}
\newcommand{\Casesumy}[2]{\matchor \Match{\Inr{#1}}{#2}}
\newcommand{\Casesum}[5]{\Case{#1}{\Match{\Inl{#2}}{#3} \Casesumy{#4}{#5}}}

\newcommand{\matchor}{\ensuremath{\normalfont\,\texttt{|}\hspace{-5.35pt}\texttt{|}\,}}
\newcommand{\Match}[2]{{#1}\Rightarrow{#2}}
\newcommand{\Let}[3]{{\mathbf{let}\:#1\,{\texttt{=}}\,#2\:\mathbf{in}\:#3}}

\newcommand{\tyname}[1]{\ensuremath{\mathsf{#1}}}
\newcommand{\keyword}[1]{\text{\usefont{T1}{\rmdefault}{b}{n}\fontsize{\keywordfontsize}{12pt}\selectfont#1}}
\newcommand{\textkw}[1]{\keyword{#1}}

\newcommand{\Bits}{\tyname{bits}}
\newcommand{\Nat}{\tyname{nat}}

\newcommand{\Pos}{\tyname{pos}}

\newcommand{\Int}{\tyname{int}}

\newcommand{\Real}{\tyname{real}}
\newcommand{\String}{\tyname{string}}

\newcommand{\Even}{\tyname{even}}
\newcommand{\Odd}{\tyname{odd}}

\newcommand{\Deref}[1]{\ensuremath{{\texttt{!}}{#1}}}
\newcommand{\Gets}{\ensuremath{\mathrel{\texttt{:=}}}}

\newcommand{\arrayenvr}[1]{\renewcommand{\arraystretch}{1} \begin{array}[t]{@{}r@{}}#1\end{array}}

\newcommand{\arrayenvbl}[1]{\renewcommand{\arraystretch}{1}  \begin{array}[b]{@{}l@{}}#1\end{array}}

\newcommand{\tabularenvl}[1]{\renewcommand{\arraystretch}{1} \begin{tabular}[t]{@{}l@{}}#1\end{tabular}}

\newcommand{\Unit}{\tyname{unit}}

\newcommand{\unit}{\texttt{()}}

\newcommand{\against}{\Leftarrow}
\newcommand{\has}{\Rightarrow}

\newcommand{\xSnd}{\keyword{proj$_2$}}
\newcommand{\xProj}[1]{\keyword{proj$_{#1}$}}
\newcommand{\Proj}[2]{\xProj{#1}\;#2}
\newcommand{\Fst}[1]{\Proj{1}{#1}}
\newcommand{\Snd}[1]{\Proj{2}{#1}}

\newcommand{\Inj}[2]{\keyword{inj$_{#1}$}\;#2}

\newcommand{\Inl}[1]{\Inj{1}{#1}}
\newcommand{\Inr}[1]{\Inj{2}{#1}}

\newcommand{\rulesep}{~~~~~~~}

\newcommand{\Dee}{\mathcal{D}}

\newcounter{codeLineCntr}

\newcommand{\out}[1]{}

\renewcommand{\phi}{\varphi}

\newcommand{\Figureref}[1]{Figure \ref{#1}}

\newcommand{\Sectionref}[1]{Section \ref{#1}}
\newcommand{\Secref}[1]{Sec.\ \ref{#1}}

\newcommand{\Theoremref}[1]{Theorem \ref{#1}}
\newcommand{\Thmref}[1]{Thm.\ \ref{#1}}

\newcommand{\Lemmaref}[1]{Lemma \ref{#1}}

\newcommand{\val}{~\mathsf{value}}

\newcommand{\Val}{\textkw{val}}

\newcommand{\hole}{\ensuremath{[\,]}}

\newcommand{\Refexp}[1]{\textbf{ref}~{#1}}
\newcommand{\Refty}[1]{{#1}~\textsf{ref}}

\newcommand{\Lbrack}{\char"5B}
\newcommand{\Rbrack}{\char"5D}
\newcommand{\Lbrace}{\char"7B}
\newcommand{\Rbrace}{\char"7D}

\definecolor{lred}{rgb}{1.0, 0.3, 0.3}
\newrgbcolor{lred}{1.0 0.3 0.3}

\newcommand{\BLACKNODE}[1]{~~\raise4pt\hbox{\psellipse[fillstyle=solid,fillcolor=black](0, 0)(6pt,6pt)}\hspace{-3pt}{\textcolor{white}{\textbf{#1}}}~\:}
\newcommand{\REDNODE}[1]{~~\raise4pt\hbox{\psellipse[fillstyle=solid,fillcolor=lred](0, 0)(6pt,6pt)}\hspace{-3pt}{\textbf{#1}}~\:}

\newlength\zzskipwidthlen

\newcommand{\bnfas}{\mathrel{::=}}
\newcommand{\bnfalt}{\mathrel{\mid}}

\ifnum\OPTIONConf=1%
\def\OPTIONLoudLabels{0}
\else
\def\OPTIONLoudLabels{1}
\fi

\newcommand{\textt}[1]{\texttt{1}}

\newdimen{\zzzpbox}

\newdimen\zzfontsz
\newcommand{\fontsz}[2]{\zzfontsz=#1%
{\fontsize{\zzfontsz}{1.2\zzfontsz}\selectfont{#2}}}

\newlength{\zzsplatboldwidth}
\newcommand{\xsplatbold}[2]{\settowidth{\zzsplatboldwidth}{{#2}}{#2}\addtolength{\zzsplatboldwidth}{-#1}\hspace{-\zzsplatboldwidth}\raisebox{#1}{{#2}}}
\newcommand{\splatbold}[1]{\xsplatbold{-0.04mm}{#1}}

\makeatletter
\def\url@MGstyle{%
\def\UrlFont{\tiny\ttfamily}%
\def\do@url@hyp{\do\-}%
\Url@do
}
\makeatother

\ifnum\OPTIONLoudLabels=1%
\newcommand{\Label}[1]{\LoudLabel{#1}}%
\newcommand{\FLabel}[1]{\label{#1}%
{\tt\scriptsize{#1}}}%
\else%
\newcommand{\Label}[1]{\label{#1}}%
\newcommand{\FLabel}[1]{\label{#1}}%
\fi

\newcommand{\derives}{\ensuremath{\mathrel{::}}}

\def \TirNameStyle #1{{#1}}
\newcommand{\Infer}[3]{\ensuremath{\inferrule*[right={\text{\strut#1}}]{{}#2\mathstrut}{{}#3\mathstrut}}}

\newcommand{\InferAnon}[2]{\infer{#2}{#1\mathstrut}}

\newcommand{\ProofCaseRule}[1]{\item \textbf{Case }\textrm{#1}: ~ }

\newcommand{\DeeBox}[3]{\WhatBox{$\Dee \derives$}{#1}{#2}{#3}}
\newcommand{\WhatBox}[4]{%
          \!\!\!\!\!\!\begin{array}[t]{l} ~\\[-5pt]
                         \text{\mbox{#1}} ~~
                         \InferAnon{#3}{#4} ~\\[4pt] \end{array}}

\newcommand{\DeeProofCaseRule}[3]{\ProofCaseRule{#1}{\DeeBox{#1}{#2}{#3}}}

\newcommand{\BeginProof}{\renewcommand{\arraystretch}{1.1} \begin{tabular}[b]{r@{}r @{} l  l}}
\newcommand{\EndProof}{\end{tabular} \renewcommand{\arraystretch}{\mydefaultarraystretch}}

\newcommand{\Pf}[4] {&$#1$ $#2$\, & $#3$ & #4 \\}

\newcommand{\ePf}[3] {\Pf{#1}{\entails\,}{#2}{#3}}

\newcommand{\proofsep}{\,\\[-0.5em]}

\newenvironment{llproof}{\BeginProof}{\EndProof}
\newcommand{\decolumnizePf}{\end{llproof} ~\\ \begin{llproof}}

\newcommand{\proofheading}[1]{}  %

\newgray{grayboxgray}{0.85}
\newcommand{\textgraybox}[1]{\psframebox[framesep=0pt,fillcolor=grayboxgray,linewidth=0.5pt]{\fbox{\parbox[s][\totalheight]{0mm}{}#1}}}

\ifnum\OPTIONConf=1
  \newcommand{\judgboxfontsize}[1]{\large #1}
\else
  \newcommand{\judgboxfontsize}[1]{\fontsz{10pt}{#1}}
\fi

\newcommand{\judgbox}[2]{%
      {\raggedright \textgraybox{\judgboxfontsize{\ensuremath{#1}}}\!\begin{tabular}[c]{l} #2 \end{tabular} %
      \ifnum\OPTIONConf=1 \\[0pt] \else \\[6pt] \fi%
}}
\newcommand{\judgboxtwelf}[3]{%
      {\raggedright \textgraybox{\judgboxfontsize{\ensuremath{#1}}}\!\begin{tabular}[c]{l} #2 \end{tabular} %
      \!\!\textgraybox{\fontsz{8pt}{#3}}%
      \ifnum\OPTIONConf=1 \\[0pt] \else \\[6pt] \fi%
}}

\newcommand{\mathcolor}[2]{\textcolor{#1}{\ensuremath{#2}}}

\newcommand{\E}{\mathcal{E}}

\newcommand{\logimp}{\mathrel{\supset}}          %

\newcommand{\Rarrelim}{\ensuremath{{\arr}\text{E}}\xspace}
\newcommand{\Rapp}{\Rarrelim}
\newcommand{\Rarrintro}{\ensuremath{{\arr}\text{I}}\xspace}
\newcommand{\Rlam}{\Rarrintro}

\newcommand{\Rfix}{\ensuremath{\mathsf{fix}}\xspace}

\newcommand{\Rsectintro}{\ensuremath{\land \text{I}}\xspace}
\newcommand{\Rtopintro}{\ensuremath{\topty \text{I}}\xspace}
\newcommand{\Rsectelim}[1]{\ensuremath{\land \text{E}_{#1}}\xspace}

\newcommand{\Runintro}[1]{\ensuremath{\lor \text{I}_{#1}}\xspace}
\newcommand{\Runelim}{\ensuremath{\lor \text{E}}\xspace}
\newcommand{\Rdirect}{\ensuremath{\mathsf{direct}}\xspace}

\newcommand{\Rvar}{\ensuremath{\mathsf{var}}\xspace}

\newcommand{\Rsub}{\ensuremath{\mathsf{sub}}\xspace}

\newcommand{\Rprodelim}[1]{\ensuremath{{*}\text{E}_{#1}}\xspace}
\newcommand{\Rfst}{\Rprodelim{1}}

\newcommand{\Rprodintro}{\ensuremath{{*}I}\xspace}

\newcommand{\MKSUB}[1]{\ensuremath{{#1}{\subtype}}\xspace}
\newcommand{\subArr}{\MKSUB{\arr}}

\newcommand{\subSectL}[1]{\MKSUB{{\sectty}\text{L}\ensuremath{_{#1}}}}
\newcommand{\xsubSectR}{\MKSUB{{\sectty}\text{R}}}
\newcommand{\subSectR}{\xsubSectR}

\newcommand{\xsubUnionL}{\MKSUB{{\unty}\text{L}}}
\newcommand{\subUnionL}{\xsubUnionL}

\newcommand{\subUnionR}[1]{\MKSUB{{\unty}\text{R}\ensuremath{_{#1}}}}

\newcommand{\subTopR}{\MKSUB{{\topty}\text{R}}}

\newcommand{\cmtbegin}{\texttt{(*}}
\newcommand{\cmtend}{\texttt{*)}}
\newcommand{\annobegin}{\cmtbegin\texttt{\Lbrack}~\,}
\newcommand{\annoend}{\,~\texttt{\Rbrack}\cmtend}

\newcommand{\mydefaultarraystretch}{1.2}

\newcommand{\Neg}{\tyname{neg}}

\newenvironment{displ}{\vspace{1pt} \begin{center} ~\!\!}{\end{center}}

\newenvironment{mathdispl}{\ifnum\OPTIONConf=1\vspace{-12pt}\else\vspace{-15pt}\fi \begin{center}\begin{mathpar} ~\!\!}{\end{mathpar}\end{center}}

\def\quofile{15}
\newcommand{\startMquotes}{\immediate\openout \quofile=\jobname.quo}
\newcommand{\MquoteM}[4]{{\small \label{#3} \begin{quotation} #1 \vspace{-0.5em} \flushright{---#2} \end{quotation}} { \immediate\write\quofile{\expandafter\csname Mquoteentry\endcsname{}{#2}{#3}{\expandafter\csname #4\endcsname}}}}
\newcommand{\Mquote}[4]{{\small \label{#3} \begin{quotation} #1 \vspace{-0.5em} \flushright{---#2} \end{quotation}} { \immediate\write\quofile{\expandafter\csname Mquoteentry\endcsname{}{#2}{#3}{#4}}}}

\newcommand{\marginnoteRmlExamples}[1]{{}}

\newcounter{zzInOinkComment}
\newcommand{\oinkkw}[1]{\ifnum\value{zzInOinkComment}=0{\usefont{T1}{cmtt}{m}{it}{\splatbold{#1}}}\else{\normalfont\textsl{#1}}\fi}

\newcommand{\lladdconj}{\mathrel{\binampersand}}

\newcommand{\RZZaddconjR}[1]{\ensuremath{{\lladdconj}\text{R}}\xspace}

\newdimen\zzlistingsize
\newdimen\zzlistingsizedefault
\zzlistingsize=\zzlistingsizedefault
\global\def\CommentCopter{0}
\newcommand{\Lstbasicstyle}{\fontsize{\zzlistingsize}{1.05\zzlistingsize}\ttfamily}
\newcommand{\keywordcopter}{\fontsize{1.0\zzlistingsize}{1.0\zzlistingsize}\bf}
\newcommand{\stupidcopter}{\if0\CommentCopter\keywordcopter\fi}
\newcommand{\commentcopter}{\def\CommentCopter{1}\fontsize{0.95\zzlistingsize}{1.0\zzlistingsize}\rmfamily\slshape}

\newlength{\zzlstwidth}
\newcommand{\setlistingsize}[1]{\zzlistingsize=#1%
\settowidth{\zzlstwidth}{{\Lstbasicstyle~}}}
\setlistingsize{\zzlistingsizedefault}

\newcommand{\secttyPottingersymbol}{\widehat{\&}}
\newcommand{\secttyPottinger}{\mathrel{\secttyPottingersymbol}}

\definecolor{dHilite}{rgb}{0.9, 0.9, 0.6}
\definecolor{dRed}{rgb}{0.65, 0.0, 0.0}
\definecolor{DRED}{rgb}{0.65, 0.0, 0.0}
\definecolor{dGreen}{rgb}{0.0, 0.65, 0.0}
\definecolor{dDkGreen}{rgb}{0.0, 0.35, 0.0}
\definecolor{dBlue}{rgb}{0.0, 0.0, 0.65}
\definecolor{dPurple}{rgb}{0.65, 0.0, 0.65}
\definecolor{dDigPurple}{rgb}{0.5, 0.0, 0.5}
\definecolor{DDIGPURPLE}{rgb}{0.5, 0.0, 0.5}  %
\definecolor{dFaint}{rgb}{0.7, 0.7, 0.7}
\definecolor{dGray}{rgb}{0.5, 0.5, 0.5}
\definecolor{dDark}{rgb}{0.2, 0.2, 0.2}
\definecolor{dAlmostBlack}{rgb}{0.1, 0.1, 0.1}

\newcommand{\tytrans}[1]{|{#1}|}

\newcommand{\doublecomma}{{,\hspace{-0.2em},\hspace{0.17em}}}
\newcommand{\Merge}[2]{{#1} \doublecomma {#2}}

\newcommand{\Rmerge}[1]{\ensuremath{\textsf{merge}_{#1}}\xspace}

\newcommand{\eltosymbol}{\hookrightarrow}
\newcommand{\eltox}[1]{\;\mathrel{\,\eltosymbol\,}{#1}}
\newcommand{\elto}[1]{\mathcolor{dBlue}{\eltox{#1}}}
\newcommand{\withcoe}[1]{\mathcolor{dDkGreen}{\;\mathrel{\,:::\,}{#1}}}

\newcommand{\zTwelf}[1]{\textcolor{black}{\texttt{\itshape #1}}}  %
\newcommand{\TwelfFile}[1]{\href{http://www.cs.queensu.ca/~jana/intcomp/#1}{\zTwelf{#1}}}
\newcommand{\TwelfFileSpecial}[2]{\href{http://www.cs.queensu.ca/~jana/#2}{\zTwelf{#1}}}

\newcommand{\Tgamma}{G}

\newcommand{\Recordty}[1]{\texttt{\Lbrace}{#1}\texttt{\Rbrace}}
\newcommand{\Fldty}[2]{\texttt{#1}\,\texttt{:}\,{#2}}
\newcommand{\Fldtysep}{\texttt{,\;}}
\newcommand{\Recordtype}[2]{\Recordty{\Fldty{#1}{#2}}}

\newcommand{\Recordex}[1]{\texttt{\Lbrace}{#1}\texttt{\Rbrace}}
\newcommand{\Fld}[2]{\texttt{#1}\texttt{=}\,{#2}}
\newcommand{\Fldsep}{\Fldtysep}
\newcommand{\Recordexp}[2]{\Recordex{\Fld{#1}{#2}}}

\newcommand{\typeoftm}[1]{typeoftm/{#1}}
\newcommand{\Typeoftm}[1]{\!\!\!\!\begin{tabular}{l} typeoftm/\\{#1}\end{tabular}\!\!}

\title{%
  Elaborating Intersection and Union Types
}

\ifnum\OPTIONConf=1     %
  \ifnum\OPTIONArxiv=0%
        \conferenceinfo{ICFP'12,} {September 9--15, 2012, Copenhagen, Denmark.}
        \CopyrightYear{2012}
        \copyrightdata{978-1-4503-1054-3/12/09}
  \else
        \authorpermission
        \permission{%
          This work is licensed under the \\
                \href{http://creativecommons.org}{Creative Commons}\;\href{http://creativecommons.org/licenses/by-nd/3.0/}{Attribution---No Derivative Works} License.
        }
        \conferenceinfo{ICFP '12,}{September 9--15, 2012, Copenhagen, Denmark.}
        \copyrightyear{2012}
        \copyrightdata{\vspace{1pt}Copyright {\sf\copyright}~2012 Jana Dunfield%
      }
   \fi
    \authorinfo{
    Jana Dunfield
   }
    {Max Planck Institute for Software Systems \\ Kaiserslautern and Saarbr\"ucken, Germany}
    {
jd169@queensu.ca
}
\else
\author{%
     Jana Dunfield%
  \\
  \small Max Planck Institute for Software Systems \\ \small Kaiserslautern and Saarbr\"ucken, Germany
}
\fi

\begin{document}
\maketitle

\vspace*{-20pt}

\begin{abstract}
  Designing and implementing typed programming languages is hard.
  Every new type system feature requires extending the metatheory
  and implementation, which are often complicated and fragile.
  To ease this process, we would like to provide general mechanisms
  that subsume many different features.

  In modern type systems, parametric polymorphism is fundamental, but
  intersection polymorphism has gained little traction in programming languages.
  Most practical intersection
  type systems have supported only \emph{refinement intersections}, which increase the expressiveness of
  types (more precise properties can be checked) without altering the expressiveness of
  terms; refinement intersections can simply be erased during compilation.
  In contrast, \emph{unrestricted} intersections increase the expressiveness of terms,
  and can be used to encode diverse language features, promising an economy of
  both theory and implementation.

  We describe a foundation for compiling unrestricted intersection and union types:
  an elaboration type system that generates ordinary $\lambda$-calculus terms.
  The key feature is a Forsythe-like merge construct.
  With this construct, not all reductions of the source program preserve types;
  however, we prove that ordinary call-by-value evaluation of the elaborated
  program corresponds to a type-preserving evaluation of the source program.

  We also describe a prototype implementation and applications of unrestricted
  intersections and unions: records, operator overloading, and simulating dynamic typing.
\end{abstract}

{\fontsize{8pt}{8.5pt}\selectfont%
{%
\\[-10pt]}
\\[-10pt]
}

\setcounter{footnote}{0}

\section{Introduction}

In type systems, parametric polymorphism is fundamental.  It enables generic 
programming; it supports parametric reasoning about programs.  Logically, it corresponds to 
universal quantification.  

\emph{Intersection} polymorphism (the intersection type $A \sectty B$) is less well appreciated.
It enables ad hoc polymorphism; it supports \emph{irregular} generic programming.
Logically, it roughly corresponds to conjunction\footnote{In our setting, this correspondence is
strong, as we will see in \Secref{sec:sectty-overview}.}. %
Not surprisingly, then, intersection is remarkably versatile. %

For both legitimate and historical reasons, intersection types have not been
used as widely as parametric polymorphism.
One of the legitimate reasons for the slow adoption of intersection types is
that no major language has them.  A restricted form of intersection,
\emph{refinement intersection}, was
realized in two extensions of SML, SML-CIDRE~\citep{DaviesThesis}
and Stardust~\citep{Dunfield07:Stardust}.
These type systems can express properties such as bitwise parity:
after refining a type $\Bits$ of bitstrings with subtypes $\Even$ (an even number of ones) and $\Odd$ (an odd number
of ones), a bitstring concatenation function can be checked against the type

\begin{displ}
  \begin{array}[t]{r@{~}l}
         &   (\Even * \Even \arr \Even) \sectty (\Odd * \Odd \arr \Even)
\\  \sectty& (\Even * \Odd \arr \Odd) \sectty (\Odd * \Even \arr \Odd)      
  \end{array}  
\end{displ}
which satisfies the refinement restriction: all the intersected types refine a single
simple type, $\Bits * \Bits \arr \Bits$.

But these systems were only typecheckers.  To \emph{compile} a program required an ordinary
Standard ML compiler.  SML-CIDRE was explicitly limited to checking refinements
of SML types, without affecting the expressiveness of terms.  In contrast,
Stardust could typecheck some kinds of programs that used
general intersection and union types, but ineffectively:
since ordinary SML compilers don't know about intersection types,
such programs could never be run.

Refinement intersections and unions increase the expressiveness of otherwise more-or-less-conventional
type systems, allowing more precise properties of programs to be
verified through typechecking.  The point is to make
fewer programs pass the typechecker; for example, a concatenation function that didn't have
the parity property expressed by its type would be rejected.
In contrast, unrestricted intersections and unions, in cooperation
with a term-level ``merge'' construct, increase the expressiveness of the term language.
For example, given primitive operations $\texttt{Int.+} : \Int * \Int \arr \Int$
and $\texttt{Real.+} : \Real * \Real \arr \Real$, we can easily define an overloaded addition
operation by writing a merge:
\[
    \Val~{\texttt{+}}~=~{\Merge{\texttt{Int.+}\,}{\,\texttt{Real.+}}}
\]
In our type system, this function \texttt{+} can be checked against the type
$(\Int * \Int \arr \Int) \sectty (\Real * \Real \arr \Real)$.

In this paper, we consider unrestricted intersection and union types.
Central to the approach is a method for elaborating programs with intersection
and union types: elaborate intersections into products, and unions into sums.
The resulting programs have no intersections and no unions, and can be compiled
using conventional means---any SML compiler will do.  The above definition of \texttt{+}
is elaborated to a pair $\Pair{\texttt{Int.+}}{\texttt{Real.+}}$; uses of \texttt{+}
on $\Int$s become first projections of \texttt{+}, while uses on $\Real$s become second
projections of \texttt{+}.

We present a three-phase design, based on this method, that supports
one of our ultimate goals: to develop simpler compilers for full-featured type systems
by encoding many features using intersections and unions.

\begin{enumerate}
\item An \emph{encoding} phase that straightforwardly rewrites the program, for example,
  turning a multi-field record type into an intersection of single-field record types,
  and multi-field records into a ``merge'' of single-field records.

\item An \emph{elaboration} phase that transforms intersections and unions into
  products and (disjoint) sums, and intersection and union introductions and eliminations (implicit
  in the source program) into their appropriate operations: tupling, projection, injection, and case analysis.

\item A \emph{compilation} phase: a conventional compiler with no support for
  intersections, unions, or the features encoded by phase 1.
\end{enumerate}

\paragraph*{Contributions:}  Phase 2 is the main contribution of this paper.
Specifically, we will:

\begin{figure}
$\hspace{15pt}$\begin{tikzpicture}
  [auto, node distance=2cm, >=stealth, %
   descr/.style= {fill=white, inner sep=2.5pt, anchor=center}
  ]
  \node (lsrc) {Program};
  \node [below of=lsrc, node distance=2.5cm] (ltar) {Result};
  \node [right of=lsrc, node distance=1.3cm] (upperleft) {$e : A$};
     \node [above of=upperleft, node distance=1.1cm] (upperleftabove) {Source language};
     \node [above of=upperleft, node distance=0.6cm] (upperlefttypelabel) {$\arr$, $\sectty$, $\unty$};
  \node [below of=upperleft, node distance=2.5cm] (lowerleft) {$v : A$};
  \node [right of=upperleft, node distance=4.1cm] (upperright) {$M : T$};
     \node [above of=upperright, node distance=1.1cm] (upperrightabove) {Target language};
     \node [above of=upperright, node distance=0.6cm] (upperrighttypelabel) {$\arr$, $*$, $+$};
  \node [right of=lowerleft, node distance=4.1cm] (lowerright) {$W : T$};
 \draw [right hook->] (upperleft)  -- node[above] {elaboration} node[below]{} (upperright);
  \draw [right hook->] (lowerleft) -- node[above] {elaboration} node[below]{} (lowerright);

  \draw [->,decorate,decoration={zigzag,segment length=4pt,amplitude=0.8pt,post=lineto,post length=3pt}]
         ($(upperleft.south)+(-7pt,-2pt)$)
                 -- node[right] {\!\!\!\!
                   \begin{tabular}{l}
                     nondeterministic \\ evaluation \\ (cbv + merge)
                   \end{tabular}}
         ($(lowerleft.north)+(-7pt,0pt)$);

  \draw [|->]
      ($(upperright.south)+(-7pt,-2pt)$)
            -- node[right] {\!\!\!%
                   \begin{tabular}{l}
                     standard \\ evaluation \\ (cbv)
                   \end{tabular}}
      ($(lowerright.north)+(-7pt,1pt)$);
\end{tikzpicture}
\caption{Elaboration and computation}
\Label{fig:diagram}
\end{figure}

\begin{itemize}
\item develop elaboration typing rules which, given a source expression $e$
  with unrestricted intersections and unions, and a ``merging'' construct $\Merge{e_1}{e_2}$,
  typecheck and transform the program into an ordinary $\lambda$-calculus
  term $M$ (with sums and products);
\item give a nondeterministic operational semantics ($\stepse$) for source programs containing merges,
  in which not all reductions preserve types;
\item prove a consistency (simulation) result: ordinary call-by-value evaluation ($\stepstm$) of the elaborated
  program produces a value corresponding to a value resulting from (type-preserving)
  reductions of the source program---that is, the diagram in \Figureref{fig:diagram} commutes;
\item describe an elaborating typechecker that, by implementing the elaboration typing rules, takes
  programs written in an ML-like language, with unrestricted intersection and union types,
  and generates Standard ML programs that can be compiled with any SML compiler.
\end{itemize}

All proofs were checked using the Twelf proof assistant~\citep{Pfenning99:Twelf,TwelfWiki}
(with the termination checker silenced for a few inductive cases, where the
induction measure was nontrivial)
and are available on the web~\citep{TwelfSupp}.  For convenience, the names
of Twelf source files (\TwelfFileSpecial{.elf}{intcomp/}) are hyperlinks.

While the idea of compiling intersections to products is not new, this paper is its first
full development and practical expression.
An essential twist is the source-level merging construct $\Merge{e_1}{e_2}$, which
embodies several computationally distinct terms, which can be checked against various parts
of an intersection type, reminiscent of Forsythe~\citep{Reynolds96:Forsythe} and (more distantly)
the $\lambda\&$-calculus~\citep{Castagna95}.
Intersections can still be introduced \emph{without} this
construct; it is required only when no single term can describe the multiple behaviours
expressed by the intersection.  Remarkably, this merging construct also supports union eliminations with
two computationally distinct branches (unlike markers for union elimination in work such as
\citet{Pierce91:IntersectionsUnionsPolymorphism}).
As usual, we have no source-level intersection eliminations and no source-level union
introductions; elaboration puts all needed projections and injections
into the target program.

\paragraph*{Contents:}  In \Sectionref{sec:sectty-overview}, we give some brief background on intersection types,
discuss their introduction and elimination rules,
introduce and discuss the merge construct, and compare intersection types to product types.
\Sectionref{sec:unty-overview} gives background on union types, discusses \emph{their} introduction
and elimination rules, and shows how the merge construct is also useful for them.

\Sectionref{sec:source} has the details of the source language and its (unusual) operational
semantics, and describes a non-elaborating type system including subsumption.
\Sectionref{sec:target} presents the target language and its (entirely standard)
typing and operational semantics.  \Sectionref{sec:elaboration} gives the elaboration typing rules,
and proves several key results relating source typing, elaboration typing,
the source operational semantics, and the target operational semantics.

\Sectionref{sec:coherence} discusses a major caveat: the approach, at least in its present form,
lacks the theoretically and practically important property of coherence, because
the meaning of a target program depends on the choice of elaboration typing derivation.

\Sectionref{sec:encodings}
shows encodings of type system features into intersections and unions,
with examples that are successfully elaborated by our prototype implementation
(\Sectionref{sec:implementation}).  Related work is discussed in \Sectionref{sec:related},
and \Sectionref{sec:conclusion} concludes.

\section{Intersection Types}   \Label{sec:sectty-overview}

What is an intersection type?  The simplistic answer is that,
supposing that types describe sets of values,
$A \sectty B$ describes the intersection of the sets of values of $A$ and $B$.  That is,
$v : A \sectty B$ if $v : A$ and $v : B$.

Less simplistically, the name has been used for substantially different type constructors, though all have a conjunctive flavour.
The intersection type in this paper is commutative ($A \sectty B = B \sectty A$) and idempotent ($A \sectty A = A$),
following several seminal papers on intersection types~\citep{Pottinger80,CDV81:IntersectionTypes} and
more recent work with refinement intersections~\citep{Freeman91,Davies00icfpIntersectionEffects,Dunfield03:IntersectionsUnionsCBV}.
Other lines of research have worked with nonlinear and/or ordered intersections, e.g.\ \citet{Kfoury04:PrincipalityExpansion},
which seem less directly applicable to practical type systems~\citep{MN04:againstSystemI}.

For this paper, then: What is a commutative and idempotent intersection type?

One approach to this question is through the Curry-Howard correspondence.
Naively, intersection should correspond to logical conjunction---but products correspond
to logical conjunction, and intersections are not products, as is evident from comparing the
standard\footnote{For impure call-by-value languages like ML,
\Rsectintro ordinarily needs to be restricted to type a value $v$, for reasons analogous to the value restriction
on parametric polymorphism~\citep{Davies00icfpIntersectionEffects}.  Our setting, however, is not ordinary:
the technique of elaboration makes the more permissive rule safe, though user-unfriendly.  See \Sectionref{sec:value-restriction}.}
introduction and elimination rules for intersection to the (utterly standard)
rules for product.  (Throughout this paper, $k$ is existentially quantified over
$\{1, 2\}$; technically, and in the Twelf formulation, we have two rules
\Rsectelim{1} and \Rsectelim{2}, etc.)

\begin{mathdispl}
    \Infer{\Rsectintro}
          {e : A_1 \\ e : A_2}
          {e : A_1 \sectty A_2}
    \and
    \Infer{\Rsectelim{k}}
          {e : A_1 \sectty A_2}
          {e : A_k}
    \vspace{-4pt}
    \\
    \Infer{\Rprodintro}
         {e_1 : A_1  \\  e_2 : A_2}
         {\Pair{e_1}{e_2} : A_1 * A_2}
    \and
    \Infer{\Rprodelim{k}}
         {e : A_1 * A_2}
         {\Proj{k} e : A_k}
\end{mathdispl}

Here \Rsectintro types a single term $e$ which inhabits type $A_1$ \emph{and} type $A_2$:
via Curry-Howard, this means that a single proof term serves as witness to two propositions
(the interpretations of $A_1$ and $A_2$).  On the other hand, in \Rprodintro two separate
terms $e_1$ and $e_2$ witness the propositions corresponding to $A_1$ and $A_2$.
This difference was suggested by \citet{Pottinger80}, and made concrete when
\citet{Hindley84} showed that intersection (of the form described by
\citet{CDV81:IntersectionTypes} and \citet{Pottinger80})
cannot correspond to  conjunction because the following type, the intersection of the types
of the $I$ and $S$ combinators, is uninhabited:
\[
    (A \arr A) \sectty \underbrace{\big((A{\arr}B{\arr}C) \arr (A{\arr}B) \arr A \arr C\big)}_{\text{``$D$''}}
\]
yet the prospectively corresponding proposition is provable in intuitionistic logic:
\[
    (A \logimp A)
    \mathsf{~and~}
    \big((A{\logimp}B{\logimp}C) \logimp (A{\logimp}B) \logimp A\logimp C\big)
    \tag{*}
\]
Hindley notes that every term of type $A \arr A$ is $\beta$-equivalent to $e_1 = \Lam{x} x$,
and
every term of type $D$ is $\beta$-equivalent to $e_2 = \Lam{x}\Lam{y}\Lam{z} x\,z\,(y\,z)$,
the $S$ combinator.  Any term $e$ of type
$(A\,{\arr}\,A) \sectty D$ must therefore have two normal forms, $e_1$ and $e_2$,
which is impossible.

But that impossibility holds for the \emph{usual} $\lambda$-terms.  Suppose we add a \emph{merge} construct $\Merge{e_1}{e_2}$
that, quite brazenly, can step to two different things: $\Merge{e_1}{e_2} \step e_1$
and $\Merge{e_1}{e_2} \step e_2$.  Its typing rule chooses one subterm and ignores the other
(throughout this paper, the subscript $k$ ranges over $\{1, 2\}$):
\vspace{-2pt}
\begin{mathdispl}
   \Infer{\Rmerge{k}}
        {  e_k : A  }
        {\Merge{e_1}{e_2} : A}
\end{mathdispl}
In combination with \Rsectintro, the \Rmerge{k} rule allows two distinct implementations $e_1$ and $e_2$,
one for each of the components $A_1$ and $A_2$ of the intersection:
\begin{mathdispl}
  \Infer{\Rsectintro}
   {
      \Infer{\Rmerge{1}}
          {e_1 : A_1}
          {\Merge{e_1}{e_2} : A_1}
      \\
      \Infer{\Rmerge{2}}
          {e_2 : A_2}
          {\Merge{e_1}{e_2} : A_2}
    }
    {
        \Merge{e_1}{e_2} : A_1 \sectty A_2
    }
\end{mathdispl}

Now $(A \arr A) \sectty D$ \emph{is} inhabited:
\[
  \Merge{e_1}{e_2} : (A \arr A) \sectty D
\]
With this construct, the ``naive'' hope that intersection corresponds to conjunction is realized
through elaboration:
we can elaborate $\Merge{e_1}{e_2}$ to $\Pair{e_1}{e_2}$, a term of type $(A \arr A) * D$,
which does correspond to the proposition (*).  Inhabitation and provability again correspond---because we
have replaced the seemingly mysterious intersections with simple products.

For source expressions, intersection still has several properties that set it apart from product.
Unlike product, it has no elimination form.  It also lacks
an explicit introduction form; \Rsectintro is the only intro rule for $\sectty$.  While the primary purpose
of \Rmerge{k} is to derive the premises of \Rsectintro, the \Rmerge{k} rule makes no mention of
intersection (or any other type constructor).  %

\citet{Pottinger80} presents intersection $A \secttyPottinger B$ as a proposition with some evidence of $A$
that is also evidence of $B$---unlike $A \mathrel{\&} B$, corresponding to $A * B$, which has two separate pieces
of evidence for $A$ and for $B$.   In our system, though, $\Merge{e_1}{e_2}$ is a single term that
provides evidence for $A$ and $B$, so it is technically consistent with this view of intersection,
but not necessarily consistent in spirit (since $e_1$ and $e_2$ can be very different from
each other).

%
%

%
%
%
%
%
%
%

%
 %
%
%
%
%

\section{Union Types}   \Label{sec:unty-overview}

Having discussed intersection types, we can describe union types as
intersections' dual: if $v : A_1 \unty A_2$ then either $v : A_1$ or $v : A_2$
(perhaps both).  This duality shows itself in several ways.

For union $\unty$, introduction is straightforward, as elimination was straightforward
for $\sectty$ (again, $k$ is either $1$ or $2$):
\begin{mathdispl}
    \Infer{\Runintro{k}}
        {\Gamma \entails e : A_k }
        {\Gamma \entails e : A_1 \unty A_2  }
\end{mathdispl}

Coming up with a good elimination rule is trickier.  A number of appealing rules
are unsound; a sound, yet acceptably strong, rule is
\vspace{-8pt}
\begin{mathdispl}
      \Infer{\Runelim}
        {\Gamma \entails e_0 : A_1 \unty A_2
         ~~~~~~
         \arrayenvbl{\Gamma, x_1 : A_1 \entails \E[x_1] : C  
            \\
            \Gamma, x_2 : A_2 \entails \E[x_2] : C }
        }
        {\Gamma \entails \E[e_0] : C }
\end{mathdispl}

This rule types an expression $\E[e_0]$---an evaluation context $\E$ with $e_0$
in an evaluation position---where $e_0$ has the union type $A_1 \unty A_2$.
During evaluation,
$e_0$ will be some value $v_0$ such that either $v_0 : A_1$ or $v_0 : A_2$.
In the former case, the premise $x_1 : A_1 \entails \E[x_1] : C$
tells us that substituting $v_0$ for $x_1$ gives a well-typed expression $\E[v_0]$.
Similarly, the premise $x_2 : A_2 \entails \E[x_2] : C$
tells us we can safely substitute $v_0$ for $x_2$.

The restriction to a single occurrence of $e_0$ in an evaluation position
is needed for soundness in many settings---generally, in any operational
semantics in which $e_0$ might step to different expressions.  One simple example
is a function $f : (A \arr A \arr C) \sectty (B \arr B \arr C)$ and expression
$e_0 : A \unty B$, where $e_0$ changes the contents pointed to by a reference
of type $\Refty{(A \unty B)}$, before returning the new value.  The application $f\;e_0\;e_0$
would be well-typed by a rule allowing multiple occurrences of $e_0$,
but unsound: the first $e_0$ could evaluate to an $A$ and the second $e_0$
to a $B$.

The evaluation context $\E$ need not be unique, which creates some
difficulties for practical typechecking~\citep{Dunfield11:letnormal}.
For further discussion of this rule, see~\citet{Dunfield03:IntersectionsUnionsCBV}.  

We saw in \Sectionref{sec:sectty-overview} that, in the usual $\lambda$-calculus,
$\sectty$ does not correspond to conjunction; in particular, no $\lambda$-term behaves
like both the $I$ and $S$ combinators, so the intersection $(A{\arr}A) \sectty D$ (where $D$ is the
type of $S$) is uninhabited.  In our setting, though,
$(A{\arr}A) \sectty D$ \emph{is} inhabited, by the merge of $I$ and $S$.

Something similar comes up when eliminating unions.  Without the merge construct,
certain instances of union types can't be usefully eliminated.  Consider a
list whose elements have type $\Int \unty \String$.  Introducing those unions
to create the list is easy enough: use \Runintro{1} for the $\Int$s and \Runintro{2}
for the $\String$s.  Now suppose we want to print a list element $x : \Int \unty \String$,
converting the $\Int$s to their string representation and leaving the $\String$s alone.
To do this, we need a merge; for example, given a function
$g : (\Int \arr \String) \sectty (\String \arr \String)$
whose body contains a merge, use rule \Runelim on $g~x$ with $\E = g~\hole$
and $e_0 = x$.

Like intersections, unions can be tamed by elaboration.  Instead of
products, we elaborate unions to products' dual, sums (\emph{tagged} unions).
Uses of \Runintro{1} and \Runintro{2} become left and
right injections into a sum type; uses of \Runelim become ordinary case expressions.

\section{Source Language} \Label{sec:source}

\subsection{Source Syntax}

\begin{figure}[h]
  \centering

\begin{array}[t]{rcl}
\end{array}
  
  \begin{tabular}[t]{r@{~~}r@{~~}r@{~}lll}
    Source types & $A, B, C$ & $\bnfas$ &
                           $\topty
                           \bnfalt A \arr B
                           \bnfalt A \sectty B
                           \bnfalt A \unty B$
\\[4pt]
    Typing contexts & $\Gamma$ & $\bnfas$ &
                     $\cdot \bnfalt \Gamma, x : A$
\\[4pt]
    Source expressions  &   $e$ & $\bnfas$
                    &  $x \bnfalt \unit \bnfalt \Lam{x} e \bnfalt e_1\,e_2 \bnfalt \Fix{x} e$
\\  &&$\bnfalt$ & $\Merge{e_1}{e_2}$
\\[4pt]
    Source values  &   $v$ & $\bnfas$ &  $x \bnfalt \unit \bnfalt \Lam{x} e \bnfalt \Merge{v_1}{v_2}$
\\[4pt]
   Evaluation contexts  &  $\E$ & $\bnfas$
   &  $\hole \bnfalt \E\;e \bnfalt  v\;\E  \bnfalt  \Merge{\E}{e}  \bnfalt  \Merge{e}{\E}$
  \end{tabular}

  \caption{Syntax of source types, contexts and expressions}
  \label{fig:source-syntax}
\end{figure}

The source language expressions $e$ are standard, except for the feature central
to our approach, the merge $\Merge{e_1}{e_2}$.  The types $A, B, C$ are a
``top'' type $\topty$ (which will be elaborated to $\Unit$), the usual function space
$A \arr B$, intersection $A \sectty B$ and union $A \unty B$.  Values $v$ are standard,
but a merge of values $\Merge{v_1}{v_2}$ is considered a value, even though it can step!
But the step it takes is pure, in the sense that even if we incorporated (say) mutable references,
it would not interact with them.  %

\subsection{Source Operational Semantics}

\begin{figure}[htbp]
\judgboxtwelf{e \stepe e'}%
     {Source expression $e$ \\ steps to $e'$}%
     {\zTwelf{step E E'} ~in~ \TwelfFile{step.elf}}
  \begin{mathpar}
    \Infer{step/app1}
        { e_1 \stepe e_1' }
        { e_1 e_2 \stepe e_1' e_2 }
    \rulesep
    \Infer{step/app2}
        { e_2 \stepe e_2'}
        { v_1 e_2  \stepe  v_1 e_2' }
    \vspace{-6pt}
    \\
    \Infer{step/beta}
        { }
        { (\Lam{x} e) v  \stepe  [v/x]e }
    \vspace{-9pt}
    \\
    \Infer{step/fix}
         { }
         { \Fix{x} e \stepe [(\Fix{x}e) / x]e }
    \vspace{-4pt}
    \\
    \Infer{step/unmerge left}
        { }
        { \Merge{e_1}{e_2} \stepe e_1 }
    \rulesep
    \Infer{step/unmerge right}
        { }
        { \Merge{e_1}{e_2} \stepe e_2 }
    \vspace{-2pt}
    \\
    \Infer{step/merge1}
        { e_1 \stepe e_1' }
        { \Merge{e_1}{e_2} \stepe \Merge{e_1'}{e_2} }
    \rulesep
    \Infer{step/merge2}
        { e_2 \stepe e_2' }
        { \Merge{e_1}{e_2} \stepe \Merge{e_1}{e_2'} }
    \vspace{-6pt}
    \\
    \Infer{step/split}
        { }
        { e \stepe \Merge{e}{e} }
  \end{mathpar}

  \caption[source-opsem]{\!\!\tabularenvl{Source language operational semantics: \\ call-by-value + merge construct}\!\!}
  \FLabel{fig:source-opsem}
\end{figure}

The source language operational semantics (\Figureref{fig:source-opsem}) is standard (call-by-value
function application and a fixed point expression) except for the merge construct.  This peculiar
animal is a descendant of ``demonic choice'': by the `step/unmerge left' and
`step/unmerge right' rules, $\Merge{e_1}{e_2}$ can step to either $e_1$ or $e_2$.
Adding to its misbehaviours, it permits stepping within itself (`step/merge1' and `step/merge2'---note
that in `step/merge2', we don't require $e_1$ to be a value).  Worst of all, it can appear by spontaneous
fission: `step/split' turns any expression $e$ into a merge of two copies of $e$.

The merge construct makes our source language operational semantics
interesting.  It also makes it unrealistic: $\stepe$-reduction does not preserve types.
For type preservation to hold, the operational semantics would need access
to the typing derivation.  Worse, since the typing rule for merges ignores
the unused part of the merge, $\stepe$-reduction can produce expressions
that have no type at all, or are not even closed!  The point of the source operational semantics is not to
directly model computation; rather, it is a basis for checking that the elaborated
program (whose operational semantics is perfectly standard) makes sense.
We will show in \Sectionref{sec:elaboration} that, if the result $M$ of elaborating $e$
can step to some $M'$, then we can step $e \stepse e'$ where $e'$ elaborates to $M'$.

\subsection{(Source) Subtyping}

Suppose we want to pass a function $f : A \arr C$ to a function $g : ((A \sectty B) \arr C) \arr D$.
This should be possible, since $f$ requires only that its argument have type $A$; in all calls
from $g$ the argument to $f$ will also have type $B$, but $f$ won't mind.  With only the
rules discussed so far, however, the application $g~f$ is not well-typed: we can't
get inside the arrow $(A \sectty B) \arr C$.  For flexibility, we'll incorporate a subtyping system
that can conclude, for example, $A \arr C \subtype (A \sectty B) \arr C$.

The logic of the subtyping rules (\Figureref{fig:source-typing}, top) is taken straight from
\citet{Dunfield03:IntersectionsUnionsCBV}, so we only briefly give some intuition.  Roughly,
$A \subtype B$ is sound if every value of type $A$ can be treated as having type $B$.  Under a subset
interpretation, this would mean that $A \subtype B$ is justified if the set of $A$-values is a subset of
the set of $B$-values.  For example, the rule \subSectR, if interpreted set-theoretically, says that
if $A \subseteq B_1$ and $A \subseteq B_2$ then $A \subseteq (B_1 \sect B_2)$.

It is easy to show that subtyping is reflexive and transitive; see \TwelfFile{sub-refl.elf}
and \TwelfFile{sub-trans.elf}.  (Building transitivity into the structure of the rules makes
it easy to derive an algorithm; an explicit transitivity rule would have premises
$A \subtype B$ and $B \subtype C$, which involve an intermediate type $B$ that does not
appear in the conclusion $A \subtype C$.)

Having said all that, the subsequent theoretical development is
easier without subtyping.  So we will show (\Theoremref{thm:coerce})
that, given a typing derivation that uses subtyping (through the usual subsumption rule), we can
always construct a source expression of the same type that never applies the subsumption rule.
This new expression will be the same as the original one, with a few additional coercions.  For
the example above, we essentially $\eta$-expand $g~f$ to $g~(\Lam{x} f~x)$, which lets us apply
\Rsectelim{1} to $x : A \sectty B$.  Operationally, all the coercions are identities; they serve
only to ``articulate'' the type structure, making subsumption unnecessary.

Note that the coercion in rule \subUnionL is eta-expanded to allow \Runelim to eliminate
the union in the type of $x$; as discussed later, the subexpression of union type must be in
evaluation position.

\subsection{Source Typing}

\begin{figure*}[htbp]
  \centering

\judgboxtwelf{A \subtype B \withcoe{e}}%
     {Source type $A$ is a subtype of source type $B$, \\
      with coercion $e$ of type $\cdot \entails e : A \arr B$}%
     {\zTwelf{sub A B Coe CoeTyping} ~in~ \TwelfFile{typeof+sub.elf}}

  \begin{mathpar}
     \Infer{\subArr}
         { B_1 \subtype A_1 \withcoe{e}
           \\
           A_2 \subtype B_2 \withcoe{e'}
           }
         { A_1 \arr A_2 \subtype B_1 \arr B_2
           \withcoe{\Lam{f} \Lam{x} e'\;(f~(e~x))}
         }
    \and
    \Infer{\subTopR}
        {}
        { A \subtype \topty
            \withcoe
           {\Lam{x} \unit}  }
    \and
    \Infer{\subSectL{k}}
         { A_k \subtype B \withcoe {e}
         }
         { A_1 \sectty A_2 \subtype B
           \withcoe
           {e}}
    \rulesep
    \Infer{\subSectR}
         { A \subtype B_1 \withcoe {e_1}
           \\
           A \subtype B_2 \withcoe {e_2}
         }
         { A \subtype B_1 \sectty B_2
           \withcoe
           { \Merge{e_1}{e_2} }
         }
     \and
     \Infer{\subUnionL}
         { A_1 \subtype B \withcoe{e_1}
           \\
           A_2 \subtype B \withcoe{e_2}
           }
         { A_1 \unty A_2 \subtype B
             \withcoe{\Lam{x} (\Lam{y} \Merge{e_1\,y}{e_2\,y})\;x}
         }
     \rulesep
     \Infer{\subUnionR{k}}
         { A \subtype B_k \withcoe{e}
           }
         { A \subtype B_1 \unty B_2
             \withcoe{e}
         }
   \end{mathpar}
~\\

\judgboxtwelf{\Gamma \entails e : A}%
     {Source expression $e$ has source type $A$}%
     {\zTwelf{typeof+sub E A} ~in~ \TwelfFile{typeof+sub.elf}}

  \begin{mathpar}
    \Infer{\Rvar}
         {}
         {\Gamma_1, x : A, \Gamma_2 \entails x : A }
    \and
    \Infer{\Rmerge{k}}
        {\Gamma \entails   e_k : A }
        {\Gamma \entails \Merge{e_1}{e_2} : A  }
     \and
     \Infer{\Rfix}
          {\Gamma, x : A \entails e : A }
          {\Gamma \entails \Fix{x} e : A}
    \and
    \Infer{\Rtopintro}
          {}
          {\Gamma \entails v : \topty }
    \and
    \Infer{\Rlam}
          {\Gamma, x : A \entails e : B }
          {\Gamma \entails \Lam{x} e : A \arr B}
    \rulesep
    \Infer{\Rapp}
         {\Gamma \entails e_1 : A \arr B 
          \\
          \Gamma \entails e_2 : A 
         }
         {\Gamma \entails e_1\,e_2 : B }
    \\
    \Infer{\Rsectintro}
        {\Gamma \entails e : A_1    \\   \Gamma \entails e : A_2 }
        {\Gamma \entails e : A_1 \sectty A_2}
    \rulesep
    \Infer{\Rsectelim{k}}
        {\Gamma \entails e : A_1 \sectty A_2 }
        {\Gamma \entails e : A_k  }
    \vspace{-10pt}
    \\
    \Infer{\Rdirect}
        {\Gamma \entails e_0 : A  
         \\
          \Gamma, x : A \entails \E[x] : C   }
        {\Gamma \entails \E[e_0] : C}
    \and
    \Infer{\Runintro{k}}
        {\Gamma \entails e : A_k }
        {\Gamma \entails e : A_1 \unty A_2  }
    \rulesep
    \Infer{\Runelim}
        {\Gamma \entails e_0 : A_1 \unty A_2
         ~~~~~~
         \arrayenvbl{\Gamma, x_1 : A_1 \entails \E[x_1] : C  
            \\
            \Gamma, x_2 : A_2 \entails \E[x_2] : C }
        }
        {\Gamma \entails \E[e_0] : C }
    \\
    \Infer{\Rsub}
        {\Gamma \entails e : A
          \\
          A \subtype B \withcoe{e_\text{coerce}}
        }
        {\Gamma \entails e : B }
   \end{mathpar}      

  \caption{Source type system, with subsumption, non-elaborating}
  \FLabel{fig:source-typing}
\end{figure*}

The source typing rules (\Figureref{fig:source-typing}) are either standard or have
already been discussed in Sections \ref{sec:sectty-overview} and \ref{sec:unty-overview},
except for \Rdirect.

The \Rdirect rule was introduced and justified in \citet{Dunfield03:IntersectionsUnionsCBV,Dunfield04:Tridirectional}.
It is a 1-ary version of \Runelim, a sort of cut: a use of the typing $e_0 : A$ within the derivation
of $\E[e_0] : C$ is replaced by a derivations of $e_0 : A$, along with a derivation of $\E[x] : C$
that assumes $x : A$.  Curiously, in this system of rules, \Rdirect is admissible: given $e_0 : A$,
use \Runintro{1} or \Runintro{2}
to conclude $e_0 : A \unty A$, then use two copies of the derivation $x : A \entails \E[x] : C$
in the premises of \Runelim ($\alpha$-converting $x$ as needed).  So why include it?
Typing using these rules is undecidable; our implementation (\Sectionref{sec:implementation})
follows a bidirectional version of them (where typechecking is decidable, given a few annotations, similar to \citet{Dunfield04:Tridirectional}),
where \Rdirect is \emph{not} admissible.  
(A side benefit is that \Rdirect and \Runelim are similar enough that it can be helpful to
do the \Rdirect case of a proof before tackling \Runelim.)

\begin{remark*}
\Theoremref{thm:coerce}, and all subsequent theorems, are proved only for expressions
that are closed under the appropriate context, even though \Rmerge{k} does not explicitly require
that the unexamined subexpression be closed; Twelf does not support proofs
about objects with unknown variables.
\end{remark*}

\begin{theorem}[Coercion]  \Label{thm:coerce}
   If $\Dee$ derives $\Gamma \entails e : B$ then
   there exists an $e'$ such that $\Dee'$ derives
   $\Gamma \entails e' : B$, where $\Dee'$ never uses rule \Rsub.
\end{theorem}
\begin{proof}
  By induction on $\Dee$.  The interesting cases are for \Rsub and \Runelim.
  In the case for \Rsub with $A \subtype B$, we show that when the coercion $e_\text{coerce}$---which always has the form
  $\Lam{x} e_0$---is applied to an expression of type $A$, we get an expression of type $B$.  For example,
  for $\subSectL{1}$ we use \Rsectelim{1}.  This shows that $e' = (\Lam{x} e_0) \; e$ has type $B$.

  For \Runelim, the premises typing $\E[x_k]$ might ``separate'', say if the first includes subsumption
  (yielding the same $\E[x_1]$) and the second doesn't.  Furthermore, inserting coercions could break
  evaluation positions: given $\E = f~\hole$, replacing $f$ with an application $(e_\text{coerce}~f)$
  means that $\hole$ is no longer in evaluation position.  To handle these issues, let
  $e' = (\Lam{y} \Merge{e_1'}{e_2'})~e_0'$, where $e_0'$ comes from applying the induction
  hypothesis to the derivation of $\Gamma \entails e_0 : A_1 \unty A_2$, and $e_1'$ and $e_2'$ come
  from applying the induction hypothesis to the other two premises.  Now $e_0'$ \emph{is} in evaluation
  position, because it follows a $\lambda$; the \Rmerge{k} typing rule will choose the correct branch.

  For details, see \TwelfFile{coerce.elf}.  We actually encode the typings for $e_\text{coerce}$
  as hypothetical derivations in the subtyping judgment itself (\TwelfFile{typeof+sub.elf}), making the
  \Rsub case here trivial.
\end{proof}

\section{Target Language} \Label{sec:target}

Our target language is just the simply-typed call-by-value $\lambda$-calculus
extended with fixed point expressions, products, and sums.

\subsection{Target Syntax}

\begin{figure}[htbp]
  \centering

  \begin{tabular}[t]{rr@{~~}r@{~}lll}
    Target types & $T$ & $\bnfas$ &
                           $\Unit
                           \bnfalt T \arr T
                           \bnfalt T * T
                           \bnfalt T + T$
\\[4pt]
    Typing contexts & $\Tgamma$ & $\bnfas$ &
                          $\cdot  \bnfalt  \Tgamma, x : T$
\\[5pt]
    Target terms  &   $\!\!\!\!M, N$ & $\bnfas$
          &  $x \bnfalt \unit \bnfalt \Lam{x} M \bnfalt M\,N \bnfalt \Fix{x} M$
   \\ && $\bnfalt$ & $\Pair{M_1}{M_2}
                          \bnfalt \Proj{k}{M}$
   \\ && $\bnfalt$ & $\Inj{k}{M}
                     \bnfalt \arrayenvr{
                                                  \Casesumx{M}{x_1}{N_1}
                                                \\
                                                  \Casesumy{x_2}{N_2}}$                          
\\[16pt]
    Target values  &   $W$ & $\bnfas$
          &  $x \bnfalt \unit \bnfalt \Lam{x} M \bnfalt \Pair{W_1}{W_2} \bnfalt \Inj{k}{W}$
  \end{tabular}
  
  \caption{Target types and terms}
  \label{fig:target-syntax}
\end{figure}

The target types and terms (\Figureref{fig:target-syntax}) are completely standard.

\subsection{Target Typing}

The typing rules for the target language (\Figureref{fig:target-typing})
lack any form of subtyping, and are completely standard.

\begin{figure*}[htbp]
  \centering
  
  \judgboxtwelf{\Tgamma \entails M : T}%
      {Target term $M$ has target type $T$}%
      {\zTwelf{typeoftm M T} ~in~ \TwelfFile{typeoftm.elf}}

      \begin{mathpar}
        \Infer{\Typeoftm{var}}
            {}
            {\Tgamma_1, x : T, \Tgamma_2 \entails x : T}
        \and
        \Infer{\Typeoftm{fix}}
            {\Tgamma, x : T \entails M : T}
            {\Tgamma \entails \Fix{x} M : T}
        \and
        \Infer{\Typeoftm{unitintro}}
             {}
             {\Tgamma \entails \unit : \Unit}
        \\
        \Infer{\Typeoftm{arrintro}}
             {\Tgamma, x : T_1 \entails M : T_2}
             {\Tgamma \entails \Lam{x} M : (T_1 \arr T_2)}
         \rulesep
         \Infer{\Typeoftm{arrelim}}
             {\Tgamma \entails  M_1 : T \arr T'
               \\
               \Tgamma \entails  M_2 : T}
             { \Tgamma \entails  M_1\,M_2 : T'}
          \\
          \Infer{\Typeoftm{prodintro}}
              {\Tgamma \entails M_1 : T_1
                \\
                \Tgamma \entails M_2 : T_2}
              {\Tgamma \entails \Pair{M_1}{M_2} : (T_1 * T_2)}
           \rulesep
           \Infer{\Typeoftm{prodelim${}_k$}}
                {\Tgamma \entails M : (T_1 * T_2)}
                {\Tgamma \entails (\Proj{k}{M}) : T_k}
            \\
           \Infer{\Typeoftm{sumintro${}_k$}}
                {\Tgamma \entails M : T_k}
                {\Tgamma \entails (\Inj{k}{M}) : (T_1 + T_2)}
            \rulesep
            \Infer{\Typeoftm{sumelim}}
                 {\Tgamma \entails M : T_1 + T_2
                    \\
                    \arrayenvbl{\Tgamma, x_1 : T_1 \entails  N_1 : T
                        \\
                        \Tgamma, x_2 : T_2 \entails  N_2 : T
                     }
                 }
                 {\Tgamma \entails (\Casesum{M}{x_1}{N_1}{x_2}{N_2}) :  T}
      \end{mathpar}

  \caption{Target type system with functions, products and sums}
  \label{fig:target-typing}
\end{figure*}

\subsection{Target Operational Semantics}

The operational semantics $M \steptm M'$ is, likewise, standard;
functions are call-by-value and products are strict.  As usual, we write
$M \stepstm M'$ for a sequence of zero or more $\steptm$s.

Naturally, a type safety result holds:

\begin{theorem}[Target Type Safety]  \Label{thm:tm-safety}
  If $\cdot \entails M : T$
  then either $M$ is a value, or $M \steptm M'$ and $\cdot \entails M' : T$.
\end{theorem}
\vspace{-11pt}
\begin{proof}
  By induction on the given derivation, using a few standard lemmas; see \TwelfFile{tm-safety.elf}.
  (The necessary substitution lemma comes for free in Twelf.)
\end{proof}

\vspace{-3pt}
And to calm any doubts about whether $M$ might step to some
other, not necessarily well-typed term:

\begin{theorem}[Determinism of $\steptm$] \Label{thm:tm-deterministic} ~\\
  If $M \steptm N_1$ and $M \steptm N_2$ then $N_1 = N_2$ (up to $\alpha$-conversion).
\end{theorem}
\vspace{-11pt}
\begin{proof}
  By simultaneous induction.  See \zTwelf{tm-deterministic} in \TwelfFile{tm-safety.elf}.
\end{proof}

\begin{figure}[htbp]
\judgboxtwelf{M \steptm M'}%
     {Target term $M$ steps to $M'$}%
     {\!\!\!\!\begin{tabular}{l}  \zTwelf{steptm M M'} \\~in~ \TwelfFile{steptm.elf} \end{tabular}\!\!\!\!}
~\\[-10pt]
  \begin{mathpar}
    \Infer{}
        { M_1 \steptm M_1' }
        { M_1 M_2 \steptm M_1' M_2 }
    \rulesep
    \Infer{}
        { M_2 \steptm M_2'}
        { W_1 M_2  \steptm  W_1 M_2' }
    \vspace{-10pt}
    \\
     \Infer{}
        { }
        { (\Lam{x} M) W  \steptm  [W/x]M }
    \rulesep
    \Infer{}
         { }
         { \Fix{x} M \steptm [(\Fix{x}M) / x]M }
    \vspace{-4pt}
    \\
    \Infer{}
        { M \steptm M'}
        { \Proj{k} {M'} \steptm \Proj{k} {M'} }
    \rulesep
    \Infer{}
        { }
        { \Proj{k} {\Pair{W_1}{W_2}} \steptm W_k }
    \vspace{-5pt}
    \\
    \Infer{}
        {M_1 \steptm M_1'}
        {\Pair{M_1}{M_2} \steptm \Pair{M_1'}{M_2}}
    \rulesep
    \Infer{}
        {M_2 \steptm M_2'}
        {\Pair{W_1}{M_2} \steptm \Pair{W_1}{M_2'}}
    \vspace{-4pt}
    \\
    \Infer{}
         {M \steptm M'}
         {\Inj{k}{M} \steptm \Inj{k}{M'}}
    \and
    \Infer{}
         {M \steptm M'}
         {\Case{M}{MS} \steptm \Case{M'}{MS}}
    \vspace{-9pt}
    \\
    \Infer{}
         {}
         {\Casesum{\Inj{k}{W}}{x_1}{N_1}{x_2}{N_2}~\steptm~[W/x_k]N_k}
    \rulesep
    \Infer{}
         {}
         {}
  \end{mathpar}
  
  \caption[target-opsem]{\!\!\tabularenvl{Target language operational semantics: \\ call-by-value + products + sums}\!\!}
  \FLabel{fig:target-opsem}
  \vspace{-6pt}
\end{figure}

\section{Elaboration Typing}  \Label{sec:elaboration}

We elaborate source expressions $e$ into target terms $M$.  The source expressions,
which include a ``merge'' construct $\Merge{e_1}{e_2}$, are typed with intersections and unions,
but the result of elaboration is completely standard and can be typed with just
$\Unit$, $\arr$, $*$ and $+$.

The elaboration judgment $\Gamma \entails e : A \elto{M}$ is read ``under assumptions $\Gamma$,
source expression $e$ has type $A$ and elaborates to target term $M$''.  
While not written explicitly in the judgment, the elaboration rules ensure
that $M$ has type $\tytrans{A}$, the \emph{type translation} of $A$ (\Figureref{fig:tytrans}).
For example,
$\tytrans{\topty \sectty (\topty{\arr}\topty)} = \Unit * (\Unit{\arr}\Unit)$.

To simplify the technical development, the elaboration rules work only
for source expressions that can be typed without using the subsumption rule \Rsub
(\Figureref{fig:source-typing}).  Such source expressions can always be produced
(\Theoremref{thm:coerce}, above).

The rest of this section discusses the elaboration rules and proves related properties:

\begin{enumerate}
\item[\ref{sec:elab:typing}] connects elaboration, source typing, and target typing;
\item[\ref{sec:elab:relation}] gives lemmas useful for showing that target computations
  correspond to source computations;
\item[\ref{sec:elab:consistency}] states and proves that correspondence (\emph{consistency},
  \Thmref{thm:consistency});
\item[\ref{sec:elab:main}] summarizes the metatheory through two important corollaries
  of our various theorems.
\end{enumerate}

Finally, \Sectionref{sec:value-restriction} discusses whether we need a value restriction on \Rsectintro.

\subsection{Connecting Elaboration and Typing}  \Label{sec:elab:typing}

\paragraph*{Equivalence of elaboration and source typing:}
$\!\!\!$The non-elaborating type assignment system of \Figureref{fig:source-typing},
minus \Rsub, can be read off from the elaboration rules in \Figureref{fig:elaboration}:
simply drop the $\elto{\dots}$ part of the judgment.  Consequently, given
$e : A \elto{M}$ we can always derive $e : A$:

\smallskip

\begin{theorem}  \Label{thm:typeof-elab} ~\\
  If $\Gamma \entails e : A \elto{M}$ then $\Gamma \entails e : A$
  (without using rule \Rsub).
\end{theorem}
\vspace{-13pt}
\begin{proof}  By straightforward induction on the given derivation;
  see \zTwelf{typeof-erase} in \TwelfFile{typeof-elab.elf}.
\end{proof}

More interestingly, given $e : A$ we can always elaborate $e$, so elaboration
is just as expressive as typing:

\smallskip

\begin{theorem}[Completeness of Elaboration]  \Label{thm:elab-complete} ~\\
  If $\Gamma \entails e : A$ (without using rule \Rsub)
  then
  $\Gamma \entails e : A \elto{M}$.
\end{theorem}
\vspace{-12pt}
\begin{proof}
  By straightforward induction on the given derivation;
  see \zTwelf{elab-complete} in \TwelfFile{typeof-elab.elf}.
\end{proof}

\paragraph*{Elaboration produces well-typed terms:}

\begin{figure}[t]
\centering

  \begin{tabular}{c}
   \begin{array}[t]{rcl}
   \tytrans{\topty} &=& \Unit
\\[2pt]
     \tytrans{A_1 \arr A_2} &=& \tytrans{A_1} \arr \tytrans{A_2}
\\[2pt]
    \tytrans{A_1 \sectty A_2} &=& \tytrans{A_1} * \tytrans{A_2}
\\[2pt]
     \tytrans{A_1 \unty A_2} &=& \tytrans{A_1} + \tytrans{A_2}
   \end{array}    
  \end{tabular}

  \caption{Type translation}
  \FLabel{fig:tytrans}
\end{figure}

Any target term $M$ produced by the elaboration rules has corresponding
target type.  In the theorem statement, we assume the obvious translation $\tytrans{\Gamma}$,
e.g.\ $\tytrans{x\,{:}\,\topty, y\,{:}\,\topty \unty \topty} =
x\,{:}\,\tytrans{\topty}, y\,{:}\,\tytrans{\topty \unty \topty} = x\,{:}\,\Unit, y\,{:}\,\Unit + \Unit$).

\smallskip

\begin{theorem}[Elaboration Type Soundness]   \Label{thm:elab-type-soundness}
  ~\\
  If $\Gamma \entails e : A \elto{M}$
  then
  $\tytrans{\Gamma} \entails M : \tytrans{A}$.
\end{theorem}
\vspace{-12pt}
\begin{proof}
  By induction on the given derivation.  For example, the case for \Rdirect,
  which elaborates to an application, applies \typeoftm{arrintro} and \typeoftm{arrelim}.
  Exploiting a bijection
  between source types and target types, we actually prove $\Gamma \entails M : A$,
  interpreting $A$ and types in $\Gamma$ as target types: $\sectty$ as $*$, etc.
  See \TwelfFile{elab-type-soundness.elf}.
\end{proof}

\begin{figure*}[htbp]
  \centering

\judgboxtwelf{\Gamma \entails e : A \elto M}%
     {Source expression $e$ has source type $A$ \\
      and elaborates to target term $M$ (of type $\tytrans{A}$)}%
     {\zTwelf{elab E A M} ~in~ \TwelfFile{elab.elf}}

  \begin{mathpar}
    \Infer{\Rvar}
         {}
         {\Gamma_1, x : A, \Gamma_2 \entails x : A \elto{x}}
    ~~~~
    \Infer{\Rmerge{k}}
        {\Gamma \entails   e_k : A \elto {M}}
        {\Gamma \entails \Merge{e_1}{e_2} : A  \elto{M}}
     \and
     \Infer{\Rfix}
          {\Gamma, x : A \entails e : A \elto {M}}
          {\Gamma \entails \Fix{x} e : A \elto {\Fix{x} M}}
    \and
    \Infer{\Rtopintro}
          {}
          {\Gamma \entails v : \topty \elto {\unit}}
    \and
    \Infer{\Rlam}
          {\Gamma, x : A \entails e : B \elto {M}}
          {\Gamma \entails \Lam{x} e : A \arr B \elto {\Lam{x} M}}
    \rulesep
    \Infer{\Rapp}
         {\Gamma \entails e_1 : A \arr B \elto {M_1}
          \\
          \Gamma \entails e_2 : A \elto {M_2}
         }
         {\Gamma \entails e_1\,e_2 : B \elto {M_1\,M_2}}
    \\
    \Infer{\Rsectintro}
        {\Gamma \entails e : A_1 \elto {M_1}   \\   \Gamma \entails e : A_2 \elto {M_2}}
        {\Gamma \entails e : A_1 \sectty A_2  \elto {\Pair{M_1}{M_2}}}
    \rulesep
    \Infer{\Rsectelim{k}}
        {\Gamma \entails e : A_1 \sectty A_2 \elto{M}}
        {\Gamma \entails e : A_k  \elto  {\Proj{k}{M}}}
    \\
    \Infer{\Runintro{k}}
        {\Gamma \entails e : A_k \elto{M}}
        {\Gamma \entails e : A_1 \unty A_2  \elto  {\Inj{k}{M}}}
    \vspace{-6pt}
    \\
    \Infer{\Rdirect}
        {\Gamma \entails e_0 : A  \elto{M_0}
         \\
          \Gamma, x : A \entails \E[x] : C \elto {N}  }
        {\Gamma \entails \E[e_0] : C  \elto { (\Lam{x} N) M_0 }}
    \rulesep
    \Infer{\Runelim}
        {\Gamma \entails e_0 : A_1 \unty A_2  \elto{M_0}
         ~~~~~~
         \arrayenvbl{\Gamma, x_1 : A_1 \entails \E[x_1] : C  \elto {N_1}
            \\
            \Gamma, x_2 : A_2 \entails \E[x_2] : C \elto {N_2}}}
        {\Gamma \entails \E[e_0] : C  \elto { \Casesum{M_0}{x_1}{N_1}{x_2}{N_2} }}
  \end{mathpar}      

  \caption{Elaboration typing rules}
  \FLabel{fig:elaboration}
\end{figure*}

\subsection{Relating Source Expressions to Target Terms} \Label{sec:elab:relation}

Elaboration produces a term that corresponds closely to the source expression:
a target term is the same as a source expression, except that the intersection- and union-related
aspects of the computation become explicit in the target.  For instance, intersection elimination via \Rsectelim{2},
implicit in the source program, becomes the explicit projection $\xSnd$.
The target term has nearly the same structure as the
source; the elaboration rules only insert operations such as $\xSnd$,
duplicate subterms such as the $e$ in \Rsectintro, and omit unused parts of merges.

This gives rise to
a relatively simple connection between source expressions and target terms---much
simpler than a logical relation, which relates all appropriately-typed terms that have
the same extensional behaviour.
In fact, stepping in the target \emph{preserves elaboration typing}, provided we are 
allowed to step the source expression zero or more times.  This consistency result,
\Theoremref{thm:consistency}, needs several lemmas.

\begin{lemma}  \Label{lem:step-eval-context}
  If $e \stepse e'$ then $\E[e] \stepse \E[e']$.
\end{lemma}
\begin{proof}  By induction on the number of steps, using a lemma (\zTwelf{step-eval-context})
  that $e \stepe e'$ implies $\E[e] \stepe \E[e']$.
  See \zTwelf{step*eval-context} in \TwelfFile{step-eval-context.elf}.
\end{proof}

Next, we prove inversion properties of unions, intersections and arrows.  Roughly, we want
to say that if an expression of union type elaborates to an injection $\Inj{k}{M_0}$, it also elaborates
to $M_0$.  For intersections, the property is slightly more complicated: given an expression
of intersection type that elaborates to a pair, we can step the expression to get something
that elaborates to the components of the pair.  Similarly, given an expression of arrow type
that elaborates to a $\lambda$-abstraction, we can step the expression to a $\lambda$-abstraction.

\begin{lemma}[Unions/Injections] ~\\  \Label{lem:elab-union}
  If $\Gamma \entails e : A_1 \unty A_2 \elto{\Inj{k}{M_0}}$
  then $\Gamma \entails e : A_k \elto{M_0}$.
\end{lemma}
\begin{proof}
  By induction on the derivation of $\Gamma \entails e : C \elto{M}$.  The
  only possible cases are \Rmerge{k} and \Runintro{k}.
  See \zTwelf{elab-inl} and \zTwelf{elab-inr} in \TwelfFile{elab-union.elf}.
\end{proof}

\begin{lemma}[Intersections/Pairs] ~\\   \Label{lem:elab-sect}
  If $\Gamma \entails e : A_1 \sectty A_2 \elto{\Pair{M_1}{M_2}}$ \\
  $~$then there exist $e_1'$ and $e_2'$ such that

  \begin{enumerate}[(1)]
  \item $e \stepse e_1'$ and $\Gamma \entails e_1' : A_1 \elto{M_1}$, and
  \item  $e \stepse e_2'$ and $\Gamma \entails e_2' : A_2 \elto{M_2}$.
  \end{enumerate}
\end{lemma}
\vspace{-9pt}
\begin{proof}
  By induction on the given derivation; the only possible cases are \Rsectintro
  and \Rmerge.  See \TwelfFile{elab-sect.elf}.
\end{proof}

\begin{lemma}[Arrows/Lambdas]  \Label{lem:elab-arr} ~\\%
  If $\cdot \entails e : A \arr B \elto{\Lam{x} M_0}$
  then there exists $e_0$ \\ such that $e \stepe^* \Lam{x} e_0$
  and $x : A \entails e_0 : B \elto{M_0}$.
\end{lemma}
\vspace{-9pt}
\begin{proof}
  By induction on the given derivation; the only possible cases are \Rarrintro and \Rmerge.
  See \TwelfFile{elab-arr.elf}.
\end{proof}

Our last interesting lemma shows that if an expression $e$ elaborates to a target value $W$,
we can step $e$ to some value $v$ that also elaborates to $W$.

\begin{lemma}[Value monotonicity]  \Label{lem:value-mono}
  If $\Gamma \entails e : A \elto{W}$ then
  $e \stepse v$ where $\Gamma \entails v : A \elto{W}$.
\end{lemma}
\vspace{-9pt}
\begin{proof}
  By induction on the given derivation.  

  The most interesting case is for \Rsectintro, where we apply the induction hypothesis to
  each premise (yielding $v_1', v_2'$ such that $e \stepse v_1'$ and $e \stepse v_2'$),
  apply the `step/split' rule to turn $e$ into $(\Merge{e}{e})$,
  and use the `step/merge1' and `step/merge2' rules to step each part of the merge,
  yielding $\Merge{v_1'}{v_2'}$, which is a value.

  In the \Rmerge{k} case on a merge $\Merge{e_1}{e_2}$, we apply the induction
  hypothesis to $e_k$, giving $e_k \stepse v$.  By rule `step/unmerge', $\Merge{e_1}{e_2} \stepe e_k$,
  from which $\Merge{e_1}{e_2} \stepse v$.

  See \TwelfFile{value-mono.elf}.
\end{proof}

\begin{lemma}[Substitution] \Label{lem:subst-elab}
 \raggedright If $\Gamma, x : A \entails e : B \elto M$ and $\Gamma \entails v : A \elto{W}$
 then $\Gamma \entails [v/x]e : B \elto{[W/x]M}$.
\end{lemma}
\vspace{-9pt}
\begin{proof}
  By induction on the first derivation.  As usual, Twelf gives us this substitution lemma for free.
\end{proof}

\subsection{Consistency} \Label{sec:elab:consistency}

This theorem is the linchpin: given $e$ that elaborates to $M$, we can preserve
the elaboration relationship even after stepping $M$, though we may have to step
$e$ some number of times as well.  The expression $e$ and term $M$, in general,
step at different speeds:

\begin{itemize}
\item $M$ steps while $e$ doesn't---for example, if $M$ is $\Proj{1}{\Pair{W_1}{W_2}}$
  and steps to $W_1$, there is nothing to do in $e$ because the projection corresponds to
  the \emph{implicit} elimination in rule \Rsectelim{1};
\item $e$ may step \emph{more} than $M$---for example, if $e$ is $(\Merge{v_1}{v_2})\,v$ and
  $M$ is $(\Lam{x}x)\,W$, then $M$ $\beta$-reduces to $W$, but $e$ must first `step/unmerge'
  to the appropriate $v_k$, yielding $v_k\,v$, and \emph{then} apply `step/beta'.
\end{itemize}
(Note that the converse---if $e \stepe e'$ then $M \stepstm M'$---does not hold:
we could pick the wrong half of a merge and get a source expression with no
particular relation to $M$.)

\begin{theorem}[Consistency]  \Label{thm:consistency} ~\\
  If
  $\cdot \entails e : A \elto{M}$
  and $M \steptm M'$
  \\
  then there exists $e'$ such that $e \stepe^* e'$ and
  $\cdot \entails e' : A \elto{M'}$.
\end{theorem}
\begin{proof}  By induction on the derivation $\Dee$ of $\cdot \entails e : A \elto{M}$.
  We show several cases here; the full proof is in \TwelfFile{consistency.elf}.

  \begin{itemize}
      \ProofCaseRule{\Rvar, \Rtopintro, \Rlam}  Impossible because $M$ cannot step.

      \smallskip

      \DeeProofCaseRule{\Rsectintro}
           {\cdot \entails e : A_1 \elto{M_1}
             ~~&~~
             \cdot \entails e : A_2 \elto{M_2}}
           {\cdot \entails e : A_1 \sectty A_2  \elto{\Pair{M_1}{M_2}}}
           
           By inversion, either $M_1 \steptm M_1'$ or $M_2 \steptm M_2'$.
           Suppose the former (the latter is similar).  By i.h.,
           $e \stepe^* e_1'$  %
           and
           $\cdot \entails e_1' : A_1 \elto{M_1'}$.  %
           By `step/split', $e \stepe \Merge{e}{e}$.
           Repeatedly applying `step/merge1' gives
           $\Merge{e}{e} \stepe^* \Merge{e_1'}{e}$.
           
           For typing, apply \Rmerge{1} with premise $\cdot \entails e_1' : A_1 \elto{M_1'}$
           and with premise $\cdot \entails e : A_2 \elto{M_2}$.

           Finally, by \Rsectintro, we have
           $\cdot \entails \Merge{e_1'}{e} : A_1 \sectty A_2 \elto{\Pair{M_1'}{M_2}}$.

           \smallskip

      \DeeProofCaseRule{\Rsectelim{k}}
          {\cdot \entails e  : A_1 \sectty A_2 \elto{M_0}}
          {\cdot \entails     e  :  A_k  \elto{\Proj{k}{M_0}}}

          If $\Proj{k}{M_0} \steptm \Proj{k}{M_0'}$ with $M_0 \steptm M_0'$, use the i.h.\ and
          apply \Rsectelim{k}.

          If $M_0 = \Pair{W_1}{W_2}$ and $\Proj{k}{M_0} \steptm W_k$,
          use \Lemmaref{lem:elab-sect}, yielding $e \stepse e_k'$ and
          $\Gamma \entails e_k' : A_k \elto{W_k}$.

          \smallskip

      \DeeProofCaseRule{\Rmerge{k}}
           {\cdot \entails e_k : A \elto{M}}
           {\cdot \entails \Merge{e_1}{e_2}  :  A  \elto{M}}

         By i.h., $e_k \stepse e'$ and $\cdot \entails e' : A$.
         By rule `step/unmerge', $\Merge{e_1}{e_2} \stepe e_k$.
         Therefore $\Merge{e_1}{e_2} \stepse e'$.

         \smallskip

      \DeeProofCaseRule{\Rapp}
         {\cdot \entails e_1 : A{\arr}B \elto{M_1}
           ~~&~
           \cdot \entails e_2 : A \elto {M_2}}
         {\cdot \entails     e_1\;e_2  :  B   \elto{M_1\;M_2}}

         We show one of the harder subcases (\zTwelf{consistency/app/beta}
         in \TwelfFile{consistency.elf}).  In this subcase, $M_1 = \Lam{x} M_0$ and
         $M_2$ is a value, with $M_1\,M_2 \steptm [M_2/x]M_0$.  We use several
         easy lemmas about stepping; for example, \zTwelf{step*app1} says that if
         $e_1 \stepe^* e_1'$ then $e_1\,e_2 \stepe^* e_1'\,e_2$.

         \smallskip

           \begin{llproof}
             Elab1 \derives \Pf{}{}{\hspace{-52pt}\cdot \entails e_1 : A \arr B \elto{\Lam{x} M_0}}   {Subd.}
             ElabBody \derives \ePf{x:A} {e_0 : B \elto{M_0}}  {By \Lemmaref{lem:elab-arr}}
             StepsFun \derives \Pf{e_1}{\stepe^*}{\Lam{x} e_0}  {\ditto}
             StepsApp \derives \Pf{~e_1\;e_2}{\stepe^*}{(\Lam{x} e_0) e_2}  {By \zTwelf{step*app1}}
             \proofsep
             Elab2 \derives \ePf{\cdot} {e_2 : A \elto{M_2}} {Subd.}
             \Pf{}{} {M_2 \val} {Above}
             Elab2$'$ \derives \ePf{\cdot} {e_2 \stepe^* v_2} {By \Lemmaref{lem:value-mono}}
             \ePf{\cdot} {v_2 : A \elto{M_2}} {\ditto}
             \Pf {(\Lam{x} e_0) e_2} {\stepe^*} {(\Lam{x} e_0) v_2} {\!\!By \zTwelf{step*app2}}
             \Pf {e_1\;e_2} {\stepe^*} {(\Lam{x} e_0) v_2} {\!\!\!\!\!\!\!\!By \zTwelf{step*append}}
             \Pf {(\Lam{x} e_0) v_2} {\stepe~} {[v_2/x]e_0} {\!\!\!\!\!\!\!\!By `step/beta'}
             \decolumnizePf
             StepsAppBeta \derives \Pf {~e_1\,e_2} {\stepe^*} {[v_2/x]e_0} {\!\!By \zTwelf{step*snoc}}
             \proofsep
             ElabBody \derives \ePf{x : A} {e_0 : B \elto{M_0}}  {Above}
           \end{llproof}
           ~\\[-5pt]
           $\cdot \entails [v_2/x]e_0 : B \elto{[M_2/x]M_0}$~~~By \Lemmaref{lem:subst-elab} (Elab2$'$)
             \qedhere
  \end{itemize}
\end{proof}

\begin{theorem}[Multi-step Consistency]   \Label{thm:consistency-star} ~\\
  If $\cdot \entails e : A \elto{M}$
  and $M \steptms W$
  then there exists $v$ such that
  $e \stepse v$ and
  $\cdot \entails v : A \elto{W}$.
\end{theorem}
\begin{proof}
  By induction on the derivation of $M \steptms W$.

  If $M$ is some value $w$ then, by \Lemmaref{lem:value-mono},
  $e$ is some value $v$.  The source expression $e$ steps to itself in zero steps,
  so $v \stepse v$, and $\cdot \entails v : A \elto{W}$ is given ($e = v$ and $M = W$).

  Otherwise, we have $M \steptm M'$ where $M' \steptms W$. 
  We want to show $\cdot \entails e' : A \elto{M'}$, where $e \stepse e'$.
  By \Theoremref{thm:consistency}, either $\cdot \entails e : A \elto{M'}$,
  or $e \stepe e'$ and $\cdot \entails e' : A \elto{M'}$.

  \begin{itemize}
  \item If $\cdot \entails e : A \elto{M'}$, let $e' = e$, so $\cdot \entails e' : A \elto{M'}$
    and $e \stepse e'$ in zero steps.
  \item If $e \stepe e'$ and $\cdot \entails e' : A \elto{M'}$,
   we can use the i.h., showing that $e' \stepse v$ and $\cdot \entails v : A \elto{W}$.
  \end{itemize}
  See \zTwelf{consistency*} in \TwelfFile{consistency.elf}.  
\end{proof}

\subsection{Summing Up} \Label{sec:elab:main}

\begin{theorem}[Static Semantics] \Label{thm:summary-static} ~\\
  If $\cdot \entails e : A$ (using any of the rules in \Figureref{fig:source-typing})
  then there exists $e'$ such that $\cdot \entails e' : A \elto{M}$
  and $\cdot \entails M : \tytrans{A}$.
\end{theorem}
\begin{proof}
  By Theorems \ref{thm:coerce} (coercion), \ref{thm:elab-complete} (completeness
  of elaboration) and \ref{thm:elab-type-soundness} (elaboration type soundness).
\end{proof}

\begin{theorem}[Dynamic Semantics] \Label{thm:summary-dynamic} ~\\
  If $\cdot \entails e : A \elto{M}$ and $M \stepstm W$ then there is a source value $v$ such that
  $e \stepse v$ and $\cdot \entails v : A$.
\end{theorem}
\begin{proof}
  By Theorems \ref{thm:consistency-star} (multi-step consistency)
  and \ref{thm:typeof-elab}.
\end{proof}

Recalling the diagram in \Figureref{fig:diagram}, \Theoremref{thm:summary-dynamic} shows that
it commutes.

Both theorems are stated and proved in \TwelfFile{summary.elf}.  Combined with
a run of the target program ($M \stepstm W$), they show that elaborated programs
are consistent with source programs.

\subsection{The Value Restriction} \Label{sec:value-restriction}

\citet{Davies00icfpIntersectionEffects} showed that the then-standard intersection introduction
(that is, our \Rsectintro) was unsound in a call-by-value semantics in the presence of effects (specifically,
mutable references).  Here is an example (modeled on theirs).  Assume a base type
$\Nat$ with values $0, 1, 2, \dots$ and a type $\Pos$ of strictly positive naturals with values $1, 2, \dots$;
assume $\Pos \subtype \Nat$.

\begin{displ}
  \begin{tabular}[c]{l}
     $\Let{r}{(\Refexp{1}) : (\Refty{\Nat}) \sectty (\Refty{\Pos})}{}$ \\
     ~~~$r \Gets 0;$ \\
     ~~~$(\Deref{r}) : \Pos$
  \end{tabular}
\end{displ}

Using the unrestricted \Rsectintro rule, $r$ has type $(\Refty{\Nat}) \sectty (\Refty{\Pos})$;
using \Rsectelim{1} yields $r : \Refty{\Nat}$, so the write $r \Gets 0$ is well-typed;
using \Rsectelim{2} yields $r : \Refty{\Pos}$, so the read $\Deref{r}$ produces a $\Pos$.
In an unelaborated setting, this typing is unsound: $(\Refexp{1})$ creates a single cell, initially containing
$1$, then overwritten with $0$, so $\Deref{r} \stepe 0$, which does not have type $\Pos$.

Davies and Pfenning proposed, analogously to ML's value restriction on $\forall$-introduction, 
an $\sectty$-introduction rule that only types values $v$.   This rule is sound with
mutable references:

\begin{mathdispl}
  \Infer{\Rsectintro\textit{ (Davies and Pfenning)}}
      {v : A_1   \\   v : A_2}
      {v : A_1 \sectty A_2}
\end{mathdispl}

In an elaboration system like ours, however, the problematic example above is sound,
because our \Rsectintro elaborates $\Refexp{1}$ to two distinct expressions, which
create two unaliased cells:

\begin{mathdispl}
  \Infer{\Rsectintro}
      {\Refexp{1} : \Refty{\Nat} \elto{\Refexp{1}}
      \\
       \Refexp{1} : \Refty{\Pos} \elto{\Refexp{1}}
      }
      {\Refexp{1} : \Refty{\Nat} \sectty \Refty{\Pos} \elto{\Pair{\Refexp{1}}{\Refexp{1}}}}
\end{mathdispl}
Thus, the example elaborates to

\begin{displ}
  \begin{tabular}[c]{l}
     $\Let{r}{\Pair{\Refexp{1}}{\Refexp{1}}}{}$ \\
     ~~~$(\Fst{r}) \Gets 0;$ \\
     ~~~$(\Deref{\Snd{r}}) : \Pos$
  \end{tabular}
\end{displ}
which is well-typed, but does not ``go wrong'' in the type-safety sense: the assignment writes to the
first cell (\Rsectelim{1}), and the dereference reads the second cell (\Rsectelim{2}), which still
contains the original value $1$.  The restriction-free \Rsectintro
thus appears sound in our setting.  Being \emph{sound} is not the same as being \emph{useful}, though;
such behaviour is less than intuitive, as we discuss in the next section. %

\section{Coherence} \Label{sec:coherence}

The merge construct, while simple and powerful, has serious usability issues when
the parts of the merge have overlapping types.  Or, more accurately, when they
would have overlapping types---types with nonempty intersection---in a merge-free system:
in our system, \emph{all} intersections $A \sectty B$ of nonempty $A$, $B$ are
nonempty: if $v_A : A$ and $v_B : B$ then $\Merge{v_A}{v_B} : A \sectty B$
by \Rmerge{k} and \Rsectintro.

According to the elaboration rules, $\Merge{0}{1}$ (checked against $\Nat$) could
elaborate to either $0$ or $1$.  Our implementation would elaborate $\Merge{0}{1}$
to $0$, because it tries the left part $0$ first.  Arguably, this is better behaviour than
actual randomness, but hardly helpful to the programmer.
Perhaps even more confusingly, suppose we are checking
$\Merge{0}{1}$ against $\Pos \sectty \Nat$, where $\Pos$ and $\Nat$ are
as in \Sectionref{sec:value-restriction}.  Our implementation would elaborate
$\Merge{0}{1}$ to $\Pair{1}{0}$, but $\Merge{1}{0}$ to $\Pair{1}{1}$.

Since the behaviour of the target program depends on the particular elaboration
typing used, the system lacks \emph{coherence}~\citep{Reynolds91:coherence}.

To recover a coherent semantics, we could limit merges according to their surface syntax,
as Reynolds did in Forsythe, but this seems restrictive; also,
crafting an appropriate syntactic restriction depends on details of the type system,
which is not robust as the type system is extended.
A more general approach might be to reject (or warn about) merges in which more
than one part checks against the same type (or the same part of an intersection type).
Implementing this seems straightforward, though it would slow typechecking since
we could not skip over $e_2$ when $e_1$ checks in $\Merge{e_1}{e_2}$.

Leaving merges aside, the mere fact that \Rsectintro elaborates the expression
twice creates problems with mutable references, as we saw in \Sectionref{sec:value-restriction}.
For this, we could revive the value restriction in \Rsectintro, at least for expressions
whose types might overlap.

\section{Applying Intersections and Unions}
  \Label{sec:encodings}

\subsection{Overloading} \Label{sec:overload}

\begin{figure}[t]
   \lstinputlisting[name=overload.rml, language=StardustML,label=lst:overload]%
                             {overload.rml}%
  
   Output of target program after elaboration:~~~\texttt{150.0; 81; 0.25}

   \medskip

  \caption{Example of overloading}
  \FLabel{fig:overloadlisting}
\end{figure}

A very simple use of unrestricted intersections is to ``overload'' operations such as
multiplication and conversion of data to printable form.  SML provides overloading
only for a fixed set of built-in operations; it is not possible to write a single
\texttt{square} function, as we do in  \Figureref{fig:overloadlisting}.  Despite
its appearance, \lstinline!$\annobegin$ val square : $\dots$ $\annoend$! is
not a comment but an annotation used to guide our bidirectional typechecker
(this syntax, inherited from Stardust, was intended for compatibility with SML
compilers, which saw these annotations as comments and ignored them).

In its present form, this idiom is less powerful than type classes~\citep{Wadler89:typeclasses}.
We could extend \texttt{toString} for lists, which would handle lists of integers
and lists of reals, but not lists of lists; the version of \texttt{toString} for lists would use
the \emph{earlier} occurrence of \texttt{toString}, defined for integers and reals
only.  Adding a mechanism for naming a type and then
``unioning'' it, recursively, is future work.

\subsection{Records} \Label{sec:records}

\citet{Reynolds96:Forsythe} developed an encoding of records using
intersection types and his version of the merge construct; similar ideas
appear in \citet{Castagna95}.
Though straightforward,
this encoding is more expressive than SML records.

The idea is to add single-field records as a primitive notion, through
a type $\Recordtype{fld}{A}$ with introduction form $\Recordexp{fld}{e}$
and the usual eliminations (explicit projection and pattern matching).
Once this is done, the multi-field record
type $\Recordty{\Fldty{fld1}{A_1}\Fldtysep\Fldty{fld2}{A_2}}$
is simply $\Recordtype{fld1}{A_1} \sectty \Recordtype{fld2}{A_2}$,
and the corresponding intro form is a merge: $\Merge{\Recordex{\Fld{fld1}{A_1}}}{\Recordex{\Fld{fld2}{A_2}}}$.
More standard concrete syntax, such as $\Recordex{\Fld{fld1}{A_1}\Fldsep\Fld{fld2}{A_2}}$,
can be handled trivially during parsing.

With subtyping on intersections, we get the desired behaviour of what SML calls
``flex records''---records with some fields not listed---with fewer
of SML's limitations.  Using this encoding, a function
that expects a record with fields $\texttt{x}$ and $\texttt{y}$ can be given
\emph{any} record that has at least those fields, whereas SML only
allows one fixed set of fields.  For example, the code in \Figureref{fig:recordlisting}
is legal in our language but not in SML.

One problem with this approach is that expressions with
duplicated field names are accepted.  This is part of the larger
issue discussed in \Sectionref{sec:coherence}.

\begin{figure}[t]
   \lstinputlisting[name=record.rml, language=StardustML,label=lst:record]%
                             {record.rml}%
  
   Output of target program after elaboration:
\begin{verbatim}
    get_xy rec1 = (1,11)
    get_xy rec2 = (2,22) (extra = 100)
    get_xy rec3 = (3,33) (other = a string)
\end{verbatim}

  \caption{Example of flexible multi-field records}
  \FLabel{fig:recordlisting}
\end{figure}

\subsection{Heterogeneous Data} \Label{sec:dynamic-typing}

A common argument for dynamic typing over static typing
is that heterogeneous data structures are more convenient.
For example, dynamic typing makes it very easy to create and manipulate lists
containing both integers and strings.  The penalty is the loss of compile-time invariant
checking.  Perhaps the lists should contain integers and strings, but not booleans;
such an invariant is not expressible in traditional dynamic typing.

A common rebuttal from advocates of static typing is that it is easy to simulate
dynamic typing in static typing.  Want a list of integers and strings?  Just
declare a datatype
\begin{lstlisting}
  datatype int_or_string = Int of int
                        | String of string
\end{lstlisting}
and use \lstinline!int_or_string list!s.  This guarantees the invariant that
the list has only integers and strings, but is unwieldy: each new element must
be wrapped in a constructor, and operations on the list elements must unwrap
the constructor, even when those operations accept both integers
and strings (such as a function of type $(\Int \arr \String) \sectty (\String \arr \String)$).

In this situation, our approach provides the compile-time invariant checking
of static typing \emph{and} the transparency of dynamic typing.  The type
of list elements (if we bother to declare it) is just a union type:
\begin{lstlisting}
  type int_union_string = int \/ string
\end{lstlisting}
Elaboration transforms programs with \lstinline!int_union_string!
into programs with \lstinline!int_or_string!.

Along these lines, we use in \Figureref{fig:dynlisting} a
type \lstinline!dyn!, defined as \lstinline!int \/ real \/ string!.
It would be useful to also allow lists, but
the current implementation lacks
recursive types of a form that could express ``\lstinline!dyn = ... \/ dyn list!''.

\begin{figure}[t]
   \lstinputlisting[name=dyn.rml, language=StardustML,label=lst:dyn]%
                             {dyn.rml}%
  
   Output of target program after elaboration:

\begin{verbatim}
    1::2::what::3.14159::4::why::nil
\end{verbatim}

  \caption{Example of heterogeneous data}
  \FLabel{fig:dynlisting}
\end{figure}

\section{Implementation} \Label{sec:implementation}

Our implementation is faithful to the spirit of the elaboration rules above, but is
substantially richer.  
It is based on Stardust, a typechecker for a subset of core
Standard ML with support for inductive datatypes, products, intersections,
unions, refinement types and indexed types~\citep{Dunfield07:Stardust},
extended with support for (first-class) polymorphism~\citep{Dunfield09}.  We do not yet
support all these features;
support for first-class polymorphism looks hardest, since Standard ML
compilers cannot even handle higher-rank predicative polymorphism.
Elaborating programs that use ML-style prenex polymorphism should work,
but we currently lack any proof or even significant testing to back that up.

Our implementation does currently support merges, intersections and unions,
a top type, a bottom (empty) type, single-field records and encoded
multi-field records (\Sectionref{sec:records}), and inductive datatypes (if their
constructors are not of intersection type, though they can take intersections
and unions as argument; removing this restriction is a high priority).

\subsection{Bidirectional Typechecking}

Our implementation uses
\emph{bidirectional typechecking}~\citep{Pierce00,Dunfield04:Tridirectional,Dunfield09},
an increasingly common technique in advanced type systems;
see \citet{Dunfield09} for references.
This technique offers two major benefits over Damas-Milner type inference:
it works for many type systems where annotation-free inference is undecidable,
and it seems to produce more localized error messages.

Bidirectional typechecking does need more type annotations.  However, by
following the approach of~\citet{Dunfield04:Tridirectional},
annotations are never needed except on redexes.  The present implementation allows
some annotations on redexes to be omitted as well.

The basic idea of bidirectional typechecking is to separate the activity of checking an
expression against a known type from the activity of synthesizing a type from the
expression itself:

\begin{displ}
  \begin{tabular}[t]{lll}
    $\Gamma \entails e \against A$  &  $e$ checks against known type $A$
\\[1pt] 
    $\Gamma \entails e \has A$  &  $e$ synthesizes type $A$
  \end{tabular}
\end{displ}

In the checking judgment, $\Gamma$, $e$ and $A$ are inputs to the typing algorithm,
which either succeeds or fails.  In the synthesis judgment, $\Gamma$ and $e$ are inputs
and $A$ is output (assuming synthesis does not fail).

Syntactically speaking, crafting a bidirectional type system from a type assignment system (like the one
in \Figureref{fig:source-typing}) is a matter of taking the colons in the $\Gamma \entails e : A$
judgments, and replacing some with ``$\against$'' and some with ``$\has$''.
Except for \Rmerge{k}, our typing rules can all be found in \citet{Dunfield04:Tridirectional},
who argued that introduction rules should check and elimination rules should synthesize.
(Parametric polymorphism muddies this picture, but see \citet{Dunfield09} for an approach
used by our implementation.)
For functions, this leads to the bidirectional rules
\begin{mathdispl}
\small
      \Infer{\Rlam}
          {\Gamma, x : A \entails e \against B }
          {\Gamma \entails \Lam{x} e \against A \arr B}
      \rulesep
      \Infer{\Rapp}
           {\Gamma \entails e_1 \has A \arr B 
            \\
            \Gamma \entails e_2 \against A 
           }
           {\Gamma \entails e_1\,e_2 \has B }
\end{mathdispl}

The merge rule, however, neither introduces nor eliminates.
We implement the obvious checking rule (which, in practice, always tries to check
against $e_1$ and, if that fails, against $e_2$):
\begin{mathdispl}
       \Infer{}%
            {\Gamma \entails   e_k \against A }
            {\Gamma \entails \Merge{e_1}{e_2} \against A  }
\end{mathdispl}
Since it can be inconvenient to annotate merges, we also implement
synthesis rules, including one that can synthesize an intersection.
\begin{mathdispl}
       \Infer{}%
            {\Gamma \entails   e_k \has A }
            {\Gamma \entails \Merge{e_1}{e_2} \has A }
       \rulesep
       \Infer{}%
            {\Gamma \entails   e_1 \has A_1
              \\
              \Gamma \entails   e_2 \has A_2}
            {\Gamma \entails \Merge{e_1}{e_2} \has A_1 \sectty A_2 }
\end{mathdispl}

Given a bidirectional typing derivation, it is generally easy to show that
a corresponding type assignment exists: replace all ``$\has$''
and ``$\against$'' with ``$:$'' (and erase explicit type annotations from
the expression).

\subsection{Performance}

Intersection typechecking is PSPACE-hard~\citep{Reynolds96:Forsythe}.  In practice,
we elaborate the examples in Figures \ref{fig:overloadlisting}, \ref{fig:recordlisting} and \ref{fig:dynlisting}
in less than a second, but they are very small.
On somewhat larger examples, such as those discussed
by \citet{Dunfield07:Stardust}, the non-elaborating version of Stardust could take minutes, thanks
to heavy use of backtracking search (trying \Rsectelim{1} then \Rsectelim{2}, etc.) and
the need to check the same expression against different types (\Rsectintro) or with different
assumptions (\Runelim).  Elaboration doesn't help with this, but it shouldn't hurt by more
than a constant factor: the shapes of the derivations and the labour of backtracking
remain the same.

To scale the approach to larger programs, we will need to consider how to efficiently represent
elaborated intersections and unions.  Like the theoretical development, the implementation
has 2-way intersection and union types, so the type $A_1 \sectty A_2 \sectty A_3$ is parsed
as $(A_1 \sectty A_2) \sectty A_3$, which becomes $(A_1 * A_2) * A_3$.  A flattened
representation $A_1 * A_2 * A_3$ would be more efficient, except when the program
uses values of type $(A_1 \sectty A_2) \sectty A_3$ where values of type $A_1 \sectty A_2$
are expected; in that case, nesting the product allows the inner pair to be passed directly
with no reboxing.  Symmetry is also likely to be an issue: passing $v : A_1 \sectty A_2$
where $v : A_2 \sectty A_1$ is expected requires building a new pair.  Here, it may
be helpful to put the components of intersections into a canonical order.

The foregoing applies to unions as well---introducing a value of a three-way union
may require two injections, and so on.

\section{Related Work} \Label{sec:related}

Intersections were originally developed by
\citet{CDV81:IntersectionTypes} and \citet{Pottinger80}, among others; \citet{Hindley92} gives
a useful introduction and bibliography.  Work on union types began later~\citep{MacQueen86:IdealModel};
\citet{Barbanera95} is a key paper on type assignment for unions.  %

\paragraph{Forsythe.}~\!\!In the late 1980s\footnote{The citation year 1996 is the date of
the revised description of Forsythe; the core ideas are found in \citet{Reynolds88:Forsythe}.},
Reynolds invented Forsythe~\citep{Reynolds96:Forsythe}, the first practical programming language
based on intersection types.
In addition to an unmarked introduction rule like \Rsectintro, the Forsythe type
system includes rules for typing a construct $p_1\texttt{,} p_2$---``a construction for intersecting or `merging' meanings''~\citep[p.\ 24]{Reynolds96:Forsythe}.
Roughly analogous to $\Merge{e_1}{e_2}$, this construct is used to encode a variety of features,
but can only be used unambiguously.
For instance, a record and a function can be merged, but two functions cannot (actually they can, but
the second phrase $p_2$ overrides the first).  %
Forsythe does not have union types.

\paragraph{The $\lambda\&$-calculus.}~\!\!%
\citet{Castagna95} developed the $\lambda\&$-calculus, which has
$\&$-terms---functions whose body is a merge, and whose type is an intersection
of arrows.  In their semantics, applying a $\&$-term to some argument reduces the term to
the branch of the merge with the smallest (compatible) domain.  Suppose we have
a $\&$-term with two branches, one of type $\Nat \arr \Nat$ and one of type
$\Pos \arr \Pos$.  Applying that $\&$-term to a value of type $\Pos$ steps to the
second branch, because its domain $\Pos$ is (strictly) a subtype of $\Nat$.

Despite the presence of a merge-like construct, their work on the
$\lambda\&$-calculus is markedly different from ours: it gives a semantics
to programs directly, and uses type information to do so, whereas we elaborate
to a standard term language with no runtime type information.
In their work, terms have both \emph{compile-time types} and \emph{run-time types}
(the run-time types
become more precise as the computation continues); the semantics
of applying a $\&$-term depends on the run-time type of the argument to choose
the branch.  The choice of the \emph{smallest} compatible domain is consistent
with notions of inheritance in object-oriented programming, where a class can
override the methods of its parent.

\paragraph{Semantic subtyping.}
Following the $\lambda\&$-calculus, \citet{Frisch08} investigated a notion of
purely semantic subtyping, where the definition of subtyping arises from a model of types,
as opposed to the syntactic approach used in our system.  They support intersections,
unions, function spaces and even complement.  Their language includes
a \emph{dynamic type dispatch} which, very roughly, combines a merge with
a generalization of our union elimination.  Again, the semantics relies on
run-time type information.  %

\paragraph{Pierce's work.}~\!\!The earliest reference I know for the idea of compiling
intersection to product is \citet{PierceThesis}: ``a language with
intersection types might even provide two different object-code sequences for the two versions
of $+$ [for $\Int$ and for $\Real$]'' (p.\ 11).
Pierce also developed a language with union types, including a term-level construct to
explicitly eliminate them~\citep{Pierce91:IntersectionsUnionsPolymorphism}.
But this construct is only a marker for where to eliminate the union: it has only
one branch, so the same term must typecheck under each assumption.  Another difference is
that this construct is the only way to eliminate a union type in his system, whereas our
\Runelim is marker-free. %
Intersections, also present in his language, have no explicit introduction construct; the introduction
rule is like our \Rsectintro.

\paragraph{Flow types.}~\!\!\citet{Turbak97} and \citet{Wells02:IntersectionFlow} use intersections in
a system with flow types.  They produce programs with \emph{virtual tuples} and \emph{virtual sums},
which correspond to the tuples and sums we produce by elaboration.  However, these constructs are
internal: nothing in their work corresponds to our explicit intersection and union term
constructors, since their system is only intended to capture existing flow properties.
They do not compile the virtual constructs into the ordinary ones.

\paragraph{Heterogeneous data and dynamic typing.}~\!\!Several approaches
to combining dynamic typing's transparency and static typing's guarantees have been
investigated.  \emph{Soft typing}~\citep{Cartwright91:SoftTyping,Aiken94:SoftTyping}
adds a kind of type inference on top of dynamic typing, but provides no ironclad guarantees.
Typed Scheme~\citep{Tobin-Hochstadt08}, developed to retroactively type Scheme programs,
has a flow-sensitive type system with union types, directly supporting heterogeneous data
in the style of  \Sectionref{sec:dynamic-typing}.  Unlike soft typing,
Typed Scheme guarantees type safety and provides genuine (even first-class) polymorphism,
though programmers are expected to provide some annotations.

\paragraph{Type refinements.}~\!\!Restricting intersections and unions to
refinements of a single base type simplifies many issues, and is conservative: 
programs can be checked against refined types, then compiled normally.  This approach
has been explored for intersections~\citep{Freeman91,Davies00icfpIntersectionEffects},
and for intersections and unions~\citep{Dunfield03:IntersectionsUnionsCBV,Dunfield04:Tridirectional}.

\section{Conclusion} \Label{sec:conclusion}

We have laid a simple yet powerful foundation for compiling unrestricted
intersections and unions: elaboration into a standard functional language.
Rather than trying to directly understand the behaviours of source programs,
we describe them via their consistency with the target programs.

The most immediate challenge is coherence:  While our elaboration approach
guarantees type safety of the compiled program, the meaning of the compiled
program depends on the particular elaboration typing derivation used; the meaning
of the source program is actually implementation-defined.

One possible solution is to restrict typing of merges so that a merge has type $A$ only if
\emph{exactly one} branch has type $A$.
We could also partially revive the value restriction, giving non-values
intersection type only if (to a conservative approximation) both components
of the intersection are provably disjoint, in the sense that no merge-free
expression has both types.

Another challenge is to reconcile, in spirit and form, the unrestricted view of intersections
and unions of this paper with the refinement approach.  Elaborating a refinement
intersection like $(\Pos \arr \Neg) \sectty (\Neg \arr \Pos)$ to a pair of functions
seems pointless (unless it can somehow facilitate optimizations in the compiler).
It will probably be necessary to have ``refinement'' and ``unrestricted'' versions
of the intersection and union type constructors, at least during elaboration; it may
be feasible to hide this distinction at the source level.

%
%
%

%

\subsection*{Acknowledgments}

In 2008, Adam Megacz suggested (after I explained the idea of compiling intersection to product)
that one could use an existing ML compiler ``as a backend''.
The anonymous ICFP reviewers' suggestions have (I hope) significantly improved
the presentation.  Finally, I had useful discussions about this work with Yan Chen,
Matthew A.\ Hammer,
Scott Kilpatrick,
Neelakantan R.\ Krishnaswami,
and
Viktor Vafeiadis.

\addtolength{\bibsep}{-2.0pt}

\vspace{-2pt}

\bibliographystyle{plainnat}
\bibliography{intcomp}

\ifnum\OPTIONAppendix=1%
\clearpage

\appendix
\onecolumn

\section{Guide to the Twelf development}

This is the PDF part of the auxiliary material to the ICFP 2012 submission,
``Elaborating Intersection and Union Types''.
The rest of the auxiliary material is Twelf code, and is available on the web:

\medskip

\begin{tabular}[t]{ll}
\url{http://www.cs.queensu.ca/~jana/intcomp.tar} & tar archive
\\
\url{http://www.cs.queensu.ca/~jana/intcomp/} & browsable files
\end{tabular}

\medskip

We give an overview and briefly describe each file (mapping back to the paper).

\subsection{Overview}

All the lemmas and theorems in the paper were proved in Twelf (version 1.7.1).
The only caveat is that, to avoid the tedium of using nontrivial induction measures
(Twelf only knows about subterm ordering), we use the \texttt{\%trustme} directive
to define \texttt{pacify}, yielding a blatantly unsound induction measure;
see \TwelfFile{base.elf}.   All uses of this unsound measure can be found with
\begin{displ}
     \texttt{grep pacify *.elf}
\end{displ}
You can easily verify that in each case where \texttt{pacify} is used, the real inductive
object is smaller according to either the standard depth (maximum path length)
or weight (number of constructors, i.e.\ number of inference rules used) measures.

In any case, you will need to set the \texttt{unsafe} flag to permit the use of \texttt{\%trustme}
in the definition of \texttt{pacify}.

\subsection{Files}

\newcommand{\twelfheader}[1]{\item~{\TwelfFile{#1}}:~~~}

\begin{itemize}
\twelfheader{base.elf}
Generic definitions not specific to this paper.

\twelfheader{syntax.elf}
Source expressions \texttt{exp}, target terms \texttt{tm}, and types \texttt{ty},
covering much of Figures \ref{fig:source-syntax}, \ref{fig:target-syntax}, and \ref{fig:tytrans}.

\twelfheader{is-value.elf}%
Which source expressions are values (\Figureref{fig:source-syntax}).

\twelfheader{eval-contexts.elf}
Evaluation contexts (\Figureref{fig:source-syntax}).

\twelfheader{is-valuetm.elf}
Which target terms are values (\Figureref{fig:target-syntax}).

\twelfheader{typeof.elf}
A system of rules for a version of $\Gamma \entails e : A$ without subtyping.
This system is related to the one in \Figureref{fig:source-typing} by \Theoremref{thm:coerce} (\TwelfFile{coerce.elf}).

\twelfheader{typeof+sub.elf}
The rules for $\Gamma \entails e : A$ (\Figureref{fig:source-typing}).
Also defines subtyping $\texttt{sub\;A\;B\;Coe\;CoeTyping}$, corresponding to
$A \subtype B \elto{\texttt{Coe}}$.  In the Twelf development,
this judgment carries its own typing derivation (in the \TwelfFile{typeof.elf} system, without
subtyping) \texttt{CoeTyping}, which shows that the coercion $\texttt{Coe}$ is well-typed.

\item~{\TwelfFile{sub-refl.elf} and \TwelfFile{sub-trans.elf}}:~~~
Reflexivity and transitivity of subtyping. %

\twelfheader{coerce.elf}
\Theoremref{thm:coerce}: Given an expression well-typed in the system of \TwelfFile{typeof+sub.elf},
with full subsumption, coercions for function types can be inserted to yield an expression
well-typed in the system of \TwelfFile{typeof.elf}.  Getting rid of subsumption
makes the rest of the development easier.

\twelfheader{elab.elf} %
Elaboration rules deriving $\Gamma \entails e : A \elto{M}$ from \Figureref{fig:elaboration}.

\twelfheader{typeof-elab.elf}
Theorems \ref{thm:typeof-elab} and \ref{thm:elab-complete}.

\twelfheader{typeoftm.elf} %
The typing rules deriving $\Tgamma \entails M : T$ from \Figureref{fig:target-typing}.

\twelfheader{elab-type-soundness.elf}
\Theoremref{thm:elab-type-soundness}.

\twelfheader{step.elf}
Stepping rules $e \stepe e'$ (\Figureref{fig:source-opsem}).

\twelfheader{step-eval-context.elf} %
\Lemmaref{lem:step-eval-context} (stepping subexpressions in evaluation position).

\twelfheader{steptm.elf} %
Stepping rules $M \steptm M'$ (\Figureref{fig:target-opsem}).

\twelfheader{tm-safety.elf} %
Theorems \ref{thm:tm-safety} and \ref{thm:tm-deterministic} (target type safety and determinism).

\item~{\TwelfFile{elab-union.elf}, \TwelfFile{elab-sect.elf}, \TwelfFile{elab-arr.elf}} %
Inversion properties of elaboration for $\unty$, $\sectty$ and $\arr$ (Lemmas \ref{lem:elab-union}, \ref{lem:elab-sect}, and \ref{lem:elab-arr}).

\twelfheader{value-mono.elf} %
Value monotonicity of elaboration (Lemma \ref{lem:value-mono}).

\twelfheader{consistency.elf} %
The main consistency result (Theorem \ref{thm:consistency}) and its multi-step version (Theorem \ref{thm:consistency-star}).

\twelfheader{summary.elf} %
Theorems \ref{thm:summary-static} and \ref{thm:summary-dynamic}, which are corollaries of earlier theorems.

\end{itemize}

\fi
\end{document}